\documentclass[sigconf]{aamas} 
\renewcommand\footnotetextcopyrightpermission[1]{} 
\usepackage{balance} 

\usepackage{booktabs} 
\usepackage{tikz} 

\usepackage{algorithm}
\usepackage{algpseudocode}
\usepackage{xcolor}
\usepackage{booktabs}
\usepackage{enumitem}
\usepackage{url}
\usepackage{multicol}
\usepackage{multirow}
\usepackage{arydshln}
\usepackage{adjustbox}
\usepackage{soul}
\usepackage{bbm}

\renewcommand{\paragraph}[1]{\smallskip\noindent\textbf{#1}}
\newenvironment{rcases}
  {\left.\begin{aligned}}
  {\end{aligned}\right\rbrace}
\usepackage{pifont}

\DeclareMathOperator*{\argmax}{arg\,max}

\newenvironment{proofsketch}{%
  \begin{proof}[Proof Sketch]%
}{%
  \end{proof}%
}

\setcopyright{ifaamas}
\acmConference[AAMAS '24]{Proc.\@ of the 23rd International Conference
on Autonomous Agents and Multiagent Systems (AAMAS 2024)}{May 6 -- 10, 2024}
{Auckland, New Zealand}{N.~Alechina, V.~Dignum, M.~Dastani, J.S.~Sichman (eds.)}
\copyrightyear{2024}
\acmYear{2024}
\acmDOI{}
\acmPrice{}
\acmISBN{}

\usepackage[absolute,overlay]{textpos}
\title[Designing Redistribution Mechanisms for Reducing Transaction Fees in Blockchains]{Designing Redistribution Mechanisms for Reducing Transaction Fees in Blockchains
}

\author{Sankarshan Damle}
\authornote{Equal Contribution.}
\affiliation{
  \institution{IIIT, Hyderabad}
  \city{Hyderbad}
  \country{India}}
\email{sankarshan.damle@research.iiit.ac.in}

\author{Manisha Padala}
\authornotemark[1]
\affiliation{
  \institution{IISc, Bangalore}
  \city{Bangalore}
  \country{India}}
\email{manishap@iisc.ac.in}

\author{Sujit Gujar}
\affiliation{
  \institution{IIIT, Hyderabad}
  \city{Hyderbad}
  \country{India}}
\email{sujit.gujar@iiit.ac.in}

\begin{abstract}
Blockchains deploy Transaction Fee Mechanisms (TFMs) to determine which user transactions to include in blocks and determine their payments (i.e., transaction fees). Increasing demand and scarce block resources have led to high user transaction fees. As these blockchains are a public resource, it may be preferable to reduce these transaction fees. To this end, we introduce Transaction Fee Redistribution Mechanisms (TFRMs) -- redistributing VCG payments collected from such TFM as rebates to minimize transaction fees. Classic redistribution mechanisms (RMs) achieve this while ensuring Allocative Efficiency (AE) and User Incentive Compatibility (UIC). Our first result shows the non-triviality of applying RM in TFMs. More concretely, we prove that it is impossible to reduce transaction fees when (i) transactions that are not confirmed do not receive rebates and (ii) the miner can strategically manipulate the mechanism. Driven by this, we propose \emph{Robust} TFRM (\textsf{R-TFRM}): a mechanism that compromises on an honest miner's individual rationality to guarantee strictly positive rebates to the users. We then introduce \emph{robust} and \emph{rational} TFRM (\textsf{R}$^2$\textsf{-TFRM}) that uses trusted on-chain randomness that additionally guarantees miner's individual rationality (in expectation) and strictly positive rebates. Our results show that TFRMs provide a promising new direction for reducing transaction fees in public blockchains.
\end{abstract}


\newcommand{\BibTeX}{\rm B\kern-.05em{\sc i\kern-.025em b}\kern-.08em\TeX}

\newtheorem{definition}{Definition}

\newtheorem{theorem}{Theorem}
\newtheorem*{Lemma}{Lemma}
\newtheorem{claim}{Claim}
\newtheorem{example}{\textbf{Example}}
\newtheorem*{theorm}{Theorem}
\newtheorem*{Claim}{Claim}

\newcommand{\rtfrm}{\textsf{R-TFRM}}

\newcommand{\rrtfrm}{\textsf{R}$^2$\textsf{-TFRM}}

\begin{document}


\pagestyle{fancy}
\fancyhead{}



\maketitle 

\begin{textblock}{15}(0.35,1)
\centering
\noindent\small In the Proceedings of the 23\textsuperscript{rd} International Conference on Autonomous Agents and Multiagent Systems (AAMAS), 2024.
\end{textblock}

\section{Introduction\label{sec::intro}}

\emph{Public} blockchains have achieved mainstream prominence with Bitcoin~\cite{nakamoto2008bitcoin} and Ethereum~\cite{buterin2014next} processing $>1M$ transactions daily \cite{ychartsBitcoin,ychartsEth}. Most commonly, public blockchains comprise a cryptographically linked series of blocks. Each block may consist of several individual transactions. \emph{Miners}, tasked with block creation, add a subset of transactions from the set of outstanding transactions (referred to as \emph{mempool}). To incentivize miners to add their transaction to the block, transaction creators (henceforth \emph{users}) include a \emph{fee} as a commission. The fee absorbs the users' \emph{valuation} for their transaction being added to the block. \citet{roughgarden21} proposes \emph{transaction fee mechanisms} (TFMs) to study the strategic interaction between the miner and the users.

\paragraph{Transaction Fee Mechanism (TFM).} TFMs resemble a classic auction setting. Users place a bid to include their transactions in the block, and the miner mimics an auctioneer to select the subset, which maximizes its revenue. E.g., Bitcoin's TFM resembles a first-price auction, where the block's miner greedily adds the transactions with the highest bids. Unfortunately, an increasing demand, cryptocurrency's market volatility, and supply-demand economics have led to users' over-paying~\cite{basu2019stablefees}. E.g., \citet{messias2020blockchain} show that $30\%$ of Bitcoin fees are two orders of magnitude more than recommended.

Considering public blockchains as a shared resource, it's desirable not to impose charges for transaction confirmation. However, given the infeasibility of confirming every transaction due to resource constraints, one may prefer only to confirm transactions with higher value (pertaining to their importance to users). The absence of transaction fees could lead users to misrepresent the value of their transactions in order to secure confirmation. Therefore, this paper aims to design TFMs that minimize transaction fees while upholding other incentive-related properties. 

Clearly, the minimization of transaction fees is at odds with the miner's objective of maximizing revenue. Thus, the task of designing TFMs to minimize fees is more intricate than in the classical auction setting, primarily because, in TFMs, miners have complete control over the transactions they include in their blocks~\cite{roughgarden21}.



\noindent\textbf{Our Goal.}  We aim to design a TFM that satisfies certain game theoretic properties like (i) Allocative Efficiency (AE): confirmed transactions maximize the overall valuation,  (ii) User Incentive Compatibility (UIC): users bid their true valuation, and Individual Rationality (IR): users receive non-negative payoff. At the same time, the TFM must actively reduce transaction fees for users, thereby enhancing the blockchain's appeal. Unfortunately, from the famous Green-Laffont Impossibility Theorem~\cite{LaffontM79}, we know that it is impossible to design a TFM that is both AE and UIC and which guarantees \emph{zero net transaction fees} -- in mechanism design commonly referred to as \emph{strong budget balance}.

Given this, our objective is to design a TFM that is both AE and UIC while minimizing the transaction fees (or is \emph{weakly budget balanced)}. Motivated from~\citet{LaffontM79}, the mechanism design literature proposes the use of \textit{Groves' Redistribution Mechanism} (RM) for this purpose \cite{cavallo08,Guo12,LaffontM79}. In Groves' RM, the VCG mechanism is executed, and then the surplus money is redistributed among the users while preserving other game-theoretic properties. 

Along similar lines, this paper introduces \emph{Transaction Fee Redistribution Mechanisms} (TFRMs): a general class of TFMs based on RMs where the miner offers rebates from the transaction fees collected to the users while retaining AE, UIC, and IR. By offering users rebates, TFRMs, in effect, reduce the transaction fees paid by them.  Figure~\ref{fig::TFRM-pic} provides an overview.

%
\begin{figure*}[t]
    \centering
    \includegraphics[width=0.74\textwidth, trim={0pt 250pt 0pt 0pt},clip]{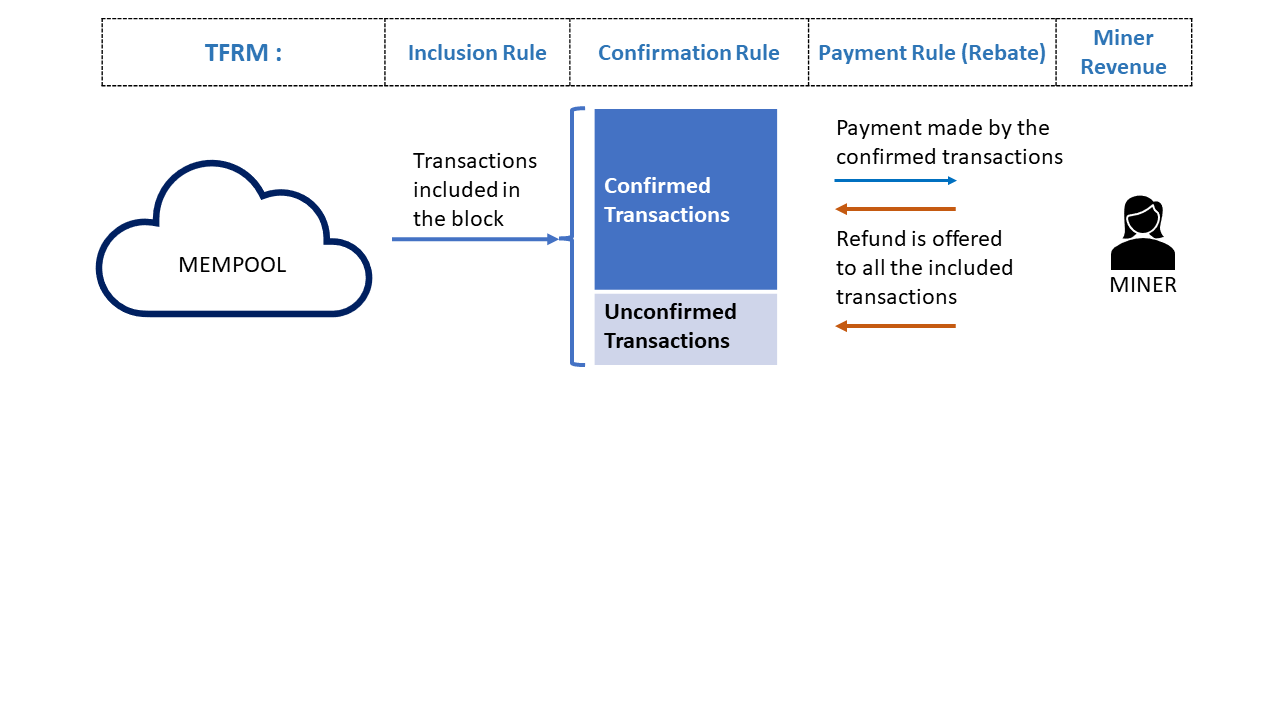}
    \caption{Overview of the framework for Transaction Fee Redistribution Mechanisms (TFRMs).}
    \label{fig::TFRM-pic}
\end{figure*}
%


\paragraph{TFRM: Challenges.} Designing such TFRMs has the following primary challenges.

\smallskip
\noindent \textit{Miner IC (MIC):} As miners possess complete control over the transactions included in their blocks~\cite{roughgarden21}, they may deviate from the intended TFM allocation rule (i.e., selecting a different subset from the mempool) and may introduce ``fake'' transactions (i.e., transactions created strategically to increase their revenue) into their blocks~\cite{roughgarden21}. This is similar to \emph{shill bidding}~\cite{porter2003cheating} in traditional auctions. Thus, it is imperative that a TFRM maintains AE, UIC, and low transaction fees even in the face of miner manipulation (or, alternately, in the presence of a \emph{strategic miner}).  

Roughgarden \cite{roughgarden21} introduces the notion of \emph{miner IC} (MIC) to model the miner's strategic behavior.  In auction theory, the Myerson-Satterthwaite impossibility theorem~\cite{myerson1983efficient} states that it is impossible to design a mechanism that is AE, IR, weakly budget balanced, and IC for both sides of the market. Designing TFRMs is analogous to such a two-sided auction, and achieving both sides' IC (with other properties) is elusive. 

\smallskip
 \noindent\textit{User IC (UIC):} Typically, RMs ensure UIC by offering rebates to everyone participating in the auction (irrespective of the allocation). In TFRMs, the transactions that are only part of the mempool (i.e., are not part of the block) are not available to the blockchain. Thus, unlike RMs, in TFRMs, we cannot offer rebates for each available transaction. As some transactions do not receive rebates, we can easily construct instances where the users of these transactions have an incentive to overbid to get included in the block and receive rebates. Thus, ensuring UIC in TFRM is non-trivial. As such, we propose \emph{restricted UIC} (RUIC), which ensures that bidding truthfully is a weakly dominant strategy \emph{only} for the users whose transactions are included in the block.

\paragraph{Our Contributions.} Broadly, we (i) formally introduce TFRMs (refer to Figure~\ref{fig::TFRM-pic} for an overview), (ii) analyze the challenges due to miner manipulation in vanilla-TFRMs, and (iii) introduce two novel TFRMs, namely \rtfrm~ and \rrtfrm~ that are robust to miner manipulation. We discuss these in detail next.


\begin{enumerate}[leftmargin=*,itemsep=0.2em]
    \item {\texttt{Ideal-TFRM}.} As we cannot offer rebates to all transactions in the mempool, we begin our analysis with an ``\texttt{Ideal-TFRM}" that offers non-zero rebates only to confirmed transactions. Unfortunately, we show that it is impossible for \texttt{Ideal-TFRM} to satisfy UIC while offering non-zero rebates to confirmed transactions~(Theorem~\ref{thm::ideal_tfrm_imp}).

\item \textbf{TFRM: Effect of A Strategic Miner.} We shift our focus to TFRMs that provide rebates to all transactions included in the block. An RM's effectiveness is measured using the Redistribution Index (RI) \cite{Guo07}, which is the fraction of the VCG surplus redistributed. To absorb the effect of strategic miners, we introduce \emph{Resilient Redistribution Index} (RRI). RRI measures the fraction of redistributed funds under optimal miner manipulation. We prove that it is impossible to design a TFRM that satisfies AE, RUIC, and is IR for both users (IR$_u$) and miners (IR$_\mathbf{M}$), while guaranteeing strictly positive RRI (Theorem~\ref{thm::imp_RRI}).

\item \textbf{Robust TFRM~(\rtfrm).} Given these impossibilities, we propose \rtfrm: a TFRM that guarantees strictly positive RRI \emph{and} satisfies all user-specific properties. However, \rtfrm\ is not individually rational for the miner\footnote{We remark that miners, or block proposers in general, often have alternate revenue streams (e.g., block rewards~\cite{nakamoto2008bitcoin} or attestation rewards~\cite{eth_stake_fees}). These rewards can primarily help absorb the reduction in revenue due to reduced transaction fees. They may also alleviate the lack of IR$_\mathbf{M}$ guarantee in \rtfrm. }. 

 At its core, \rtfrm\ builds on an RM with VCG payments as transaction fees and a linear rebate function. The rebate function maximizes the worst-case rebate while satisfying RUIC, IR$_u$, and  Approx-IR$_\mathbf{M}$. Designing such a rebate function is equivalent to solving the linear program in Figure~\ref{fig::nonir-tfrm} for its coefficients. We finally show that the payments are reduced by a fraction $k/n$, where $k$ transactions are confirmed out of the $n$ included in the block~(Theorem~\ref{thm:rri-nonir-tfrm}). The fraction remains the same, even with miner manipulation, i.e., RRI is also $k/n$. In other words, each confirmed user sees a reduction by $(1-k/n)$ in its transaction fee compared to the equivalent VCG-based TFM.

\item \textbf{Robust and Rational TFRM (\rrtfrm).} \rtfrm\ ensures positive RRI by compromising IR$_\mathbf{M}$. Another way of ensuring positive RRI is by randomly offering rebates to the users. Such an approach guarantees IR$_\mathbf{M}$, in expectation. \rrtfrm\ uses this approach wherein each user receives the rebate given by \rtfrm~ with probability $\alpha$ and does not receive any rebate probability $1-\alpha$. The randomization is carried out by the blockchain in a trusted manner~\cite{chung2021foundations}.   Theorem~\ref{thm:rrtfrm} shows that for $\alpha \in (0,\overline{\alpha}) \; \& \;\overline{\alpha}<1,$ \rrtfrm\ is AE, IR$_u$, and RUIC and IR$_\mathbf{M}$, in expectation. Further, it ensures an expected RRI of $\alpha\cdot k/n$.
\end{enumerate}

\section{Related Work} 
\label{sec::rw}
The role of transaction fees in decentralized cryptocurrencies such as Bitcoin and Ethereum has been studied in relation to (i) transaction latency~\cite{kelkar2020order,kursawe2020wendy,orda2021enforcing}, (ii) fairness~\cite{basu2019stablefees,bitcoinf} and most recently, as decentralized auction-based mechanisms~\cite{roughgarden21,parkes21,chung2021foundations,zhao2022bayesian}. 

\paragraph{Transaction Fee Mechanism (TFM).} \citet{roughgarden21} formulated the transaction issuer and miner interaction as an auction setting. More concretely, the author expresses popular mechanisms in the TFM framework, including first-price, second-price, and EIP-1559. \citet{roughgarden21} introduces user incentive compatibility (UIC), miner incentive compatibility (MIC), and off-chain agreement (OCA) proof as desirable incentive properties for TFMs. E.g., EIP-1559 satisfies UIC, MIC, and OCA-proof -- under specific constraints on the base fee. \citet{parkes21} present a dynamic posted-price TFM which is UIC (if the network is not congested) and MIC. The authors also provide an equilibrium posted price based on the demand. \citet{chung2021foundations} show it is impossible to construct a TFM that simultaneously satisfies UIC, MIC, and OCA-proof (even when only a single user and the miner collude). To address the impossibility, they introduce a \emph{penalty} to the creator of a fake transaction, discounted by a parameter ``$\gamma$." The authors present a randomized (based on trusted on-chain randomness) second-price auction that satisfies the three properties. Lastly, the authors in \cite{zhao2022bayesian} relax UIC to Bayesian UIC (BUIC) to construct another second-price-based auction, satisfying BUIC, MIC, and OCA-proof.

\paragraph{Redistribution Mechanism (RM).} 
Popular auction-based mechanisms like VCG and Groves~\cite{Vickrey61,Clarke71,Groves73} satisfy AE and UIC but do not satisfy SBB.
\citet{Faltings05} and \citet{Guo08} achieve SBB by compromising on AE. \citet{Hartline08} propose a mechanism that maximizes the sum of the agents' utility in expectation. \citet{Clippel14} ``destroy" some items to maximize the agents' utilities, leading to approximate AE and SBB. \citet{parkes01} propose an alternate approach by proposing an optimization problem which is approximately AE, SBB. 

\citet{LaffontM79} first propose the idea of redistribution of the surplus as far as possible after preserving UIC and AE. \citet{Bailey97}, \citet{cavalloRedis06}, \cite{Moulin09}, and \citet{Guo07} consider a setting of allocating $k$ homogeneous objects among $n$ competing agents with unit demand. \citet{Guo09} generalize their work in \cite{Guo07} to multi-unit demand to obtain worst-case optimal (WCO) RM.



In summary, the relatively recent TFM literature only focuses on the satisfiability of desirable incentive properties and \emph{not} on reducing the user cost. Given that a decentralized cryptocurrency (e.g., Bitcoin or Ethereum) is a public resource, re-imaging a TFM as an RM will (i) continue to guarantee these properties but, crucially, (ii) minimize the cost paid by the user.

\begin{table}[t]
    \centering
    \adjustbox{max width=\columnwidth}{
    \begin{tabular}{cl}
    \toprule
     \textbf{Symbol}    & \textbf{Notations} \\
    \midrule
    \multicolumn{2}{c}{TFM Model} \\
    \midrule
     $M:=\{1,\ldots,m\}$ & Mempool (set of outstanding transactions) \\
     $\mathbf{x}^I$ & Block inclusion rule \\
     $\mathbf{x}^C$  & Block confirmation rule \\
     $\mathbf{p}$ & Payment rule \\
     $\mathcal{T}:=(\mathbf{x}^I,\mathbf{x}^C,\mathbf{p})$ & TFM Tuple \\
     $\mathbf{b}:=\{b_1,\ldots,b_m\}$ & Bids present in the mempool \\
     $\Theta:=[\theta_i]$ & Set of valuations with each user $i$'s valuation $\theta_i$ \\
     $n\in\mathbb{Z}_{\geq 1}$ & Number of included transactions \\
     $k\in\mathbb{Z}:k\in[1,n]$ &  Number of confirmed transactions \\
     $\mathbf{b}_I$ & Bids of users included in the block \\
     $\mathbf{b}_{I\setminus i}$ & Bids of users included in the block without user $i$ \\
    $   u_i(\theta_i,\mathbf{b})$ & Utility of user $i$ \\
    $F$ & Set of fake transactions added by the miner \\
     $u_{\textsf{M}}(F,\mathbf{b})$ & Utility of the miner \\
         \midrule
    \multicolumn{2}{c}{TFRM Model} \\
    \midrule
    $g(\cdot)$ & Rebate Function \\
    $r_i$ & Rebate to user $i$ \\ 
    $c_i\in\mathbb{R},~\forall i \in \mathbf{b}_I$ & Rebate function constants \\
    $e_{\textsf{wc}}$ and $e_{\textsf{avg}}$ & Worst-case/Average Redistribution Index \\
    $\hat{e}_{\textsf{wc}}$ & Resilient Redistribution Index (RRI) \\
        \bottomrule   
    \end{tabular}}
    \caption{Notations}
    \label{app:tab::notations}
\end{table}

\section{Preliminaries\label{sec::prelim}}
We now (i) formally introduce TFMs, (ii) relevant game-theoretic definitions, and (iii) summarize redistribution mechanisms. Table~\ref{app:tab::notations} tabulates the notations used in this paper.

\paragraph{Transaction Fee Mechanism (TFM) Model.}
We have a strategic but myopic\footnote{We say that a miner is myopic if its utility function is its net revenue from the current block~\cite{roughgarden21,parkes21,chung2021foundations}.} miner building a block $B$ (with finite capacity) for the underlying blockchain. There are $m$ transactions available to be confirmed in a \emph{mempool} $M$. However, the block can hold only up to $n<m$ transactions. We assume that all transactions are of the same size.
Among the $n$ transactions included in the block, the miner confirms $k\leq n$ transactions\footnote{Such a setting is analogous to the `homogeneous' (unit demand) setting in the RM literature~\cite{Guo09}.}.

Let each user $i$ value the confirmation of its transaction at $\theta_i \in \mathbb{R}_{\geq 0}$. Each user $i$ submits a bid $b_i \in \mathbb{R}_{\geq 0}$. We have $\Theta := [\theta_i]$ and $\mathbf{b} := [b_i]$. Given the bid profile and valuation profile, the transaction fee mechanism is characterized by the inclusion rule, confirmation rule, and payment rule, as defined below.



\begin{definition}[Transaction Fee Mechanism (TFM)~\cite{chung2021foundations,roughgarden21}.] \label{def::tfm}
Given a bid profile $\mathbf{b}$, we define TFM as $\mathcal{T}:=(\mathbf{x}^I,\mathbf{x}^C,\mathbf{p})$ where: 
\begin{itemize}[leftmargin=*,noitemsep=0.01em]
    \item $\mathbf{x}^I$ is a feasible block inclusion rule, i.e., $\sum_{i\in M} x_i^I(\mathbf{b})\leq n$ where $x_i^I(\cdot)\in\{0,1\}$. Let the set of included transactions be $I = \{i | x^I_i = 1, i\in M \}$.
    \item $\mathbf{x}^C$ is a feasible block confirmation rule, i.e., $\sum_{i\in M}x_i^C(\mathbf{b})\leq k$ where $x_i^C(\cdot)\in\{0,1\}$. Let the set of confirmed transaction be $C = \{i: x_i^C = 1, i\in M\}$. Trivially, $C \subseteq I$.
    \item $\mathbf{p}$ is the payment rule with the payment for each included transaction, i.e., $\forall i\in \mathbf{x}^I$ the payment is denoted by $p_i(\mathbf{b},\mathbf{x}^I,\mathbf{x}^C)$.
\end{itemize}
\end{definition}

In TFMs, the included (but not confirmed) bids are often used as `price-setting' bids~\cite{chung2021foundations}. We use the example of a second-price TFM to explain Definition~\ref{def::tfm} better.

\begin{example}[Second-price TFM (SPA)~\cite{chung2021foundations,roughgarden21}.]\upshape\label{ex1}
W.l.o.g., assume that $\mathbf{b}=(b_1,\ldots,b_m)$ are bids in decreasing order. Now, the inclusion rule is $x_i^I=1,~\forall i\in \{1,\ldots,n\}$ and zero otherwise, i.e., the top $n$ transactions are included in the block. With $k = n -1$, the confirmation rule is $x_i^C=1,~\forall i\in \{1,\ldots, k\}$  and zero otherwise. The top $k$ (among $n$) transactions are confirmed, and the last included transaction is the price-setting transaction. Each confirmed user $i\in [k]$ pays $p_i=b_{k+1}$ to the miner and unconfirmed user ($i \in I \setminus C$) pays $p_i = 0$. The miner's net revenue is $k\cdot b_{k+1}$. 
\end{example}

In order to define the desirable properties of a TFM, we first define user and miner utilities.

\paragraph{Utility Model.} We first reiterate that we assume that the miners and transaction creators (or users) are myopic~\cite{roughgarden21,parkes21,chung2021foundations,zhao2022bayesian}. For each user $i$, let its transaction $i$'s valuation be $\theta_i\in\mathbb{R}_{\geq 0}$ with bid $b_i\in\mathbb{R}_{\geq 0}$. We have $\Theta:=[\theta_i]$ and $\mathbf{b}:=[b_i]$. Now, given $\mathcal{T}=(\mathbf{x}^I,\mathbf{x}^C,\mathbf{p})$, each user $i$'s \textit{quasi-linear} utility $u_i$ is defined as:
\begin{equation}\label{eqn::util}
    u_i(\theta_i,\mathbf{b}) :=  \left(\mathbbm{1}_{x_i^C=1}\cdot \theta_i \right) -p_i(\mathbf{b},\mathbf{x}^I,\mathbf{x}^C) 
\end{equation}

We now define the miner's utility. Since the miner has complete control over the transactions, it adds to its block~\cite{roughgarden21}, it can add a set of ``fake'' transactions (say $F$) to deviate from the intended allocation rule $\mathbf{x}=(\mathbf{x}^C,\mathbf{x}^I)$. The miner's utility (say $u_{\mathbf{M}}$), given $\mathcal{T}$, and for the block $B$, is given by:
\begin{equation}\label{eqn::miner-util}
    u_{\mathbf{M}}(F,\mathbf{b}) := \sum_{i\in  B\cap M} p_i(\mathbf{b},\mathbf{x}^I,\mathbf{x}^C)
\end{equation}
That is, the miner's utility only depends on the set of transactions $B\cap M$ since for transactions in $F$, it is paying to itself. 



\subsection{TFMs: Desirable Properties}\label{subsec:Tfm-dp} We now define the relevant incentive properties of a TFM (from \cite{roughgarden21,parkes21}). 
We begin by defining Individual Rationality.

\paragraph{Individual Rationality (IR).} To incentivize participation, mechanism designers also focus on IR.

\begin{definition}[(Ex-post) Individual Rationality (IR)]
Given a TFM $\mathcal{T}=(\mathbf{x}^I,\mathbf{x}^C,\mathbf{p})$, we say that it satisfies IR for both the user and miners if their utility post participation in the mechanism is non-negative, i.e., $u_i(\cdot)\geq 0, \forall i \in M$ and $u_{\mathbf{M}}(\cdot)\geq 0$.
\end{definition}

\noindent\underline{Note.} We denote a mechanism that is IR w.r.t. miner as IR$_\mathbf{M}$ and IR w.r.t. user as $\mbox{IR}_u$.

\paragraph{User Incentive Compatibility (UIC).} To provide a good user experience, TFMs must satisfy UIC, i.e., they must incentivize users to report their true valuation as their bids (or transaction fees). 

\begin{definition}[User Incentive Compatibility (UIC)~\cite{roughgarden21}] Given a TFM $\mathcal{T} =(\mathbf{x}^I,\mathbf{x}^C,\mathbf{p})$, we say that each user $i$'s strategy $b_i^\star=\theta_i$ satisfies UIC, if bidding $b_i^\star$ maximizes its utility  $u_i$ (Eq.~\ref{eqn::util}), irrespective of the bids of others. More formally, $\forall i\in M$ we have
$$
u_i(\theta_i,b_i^\star=\theta_i,\mathbf{b}_{- i}) \geq u_i(\theta_i,b_i,\mathbf{b}_{-i}), \forall \theta_i, \forall \mathbf{b}_{- i},
$$
where $\mathbf{b}_{-i}$ are the bids of all users excluding $i$.
\end{definition}
As we show later, ensuring UIC while minimizing transaction fees in a TFM is challenging. Hence, we focus on the incentive compatibility of a `restricted' set of users whose transactions are included in the block. We define restricted UIC (RUIC), which states that for all the included users, reporting truthfully is IC irrespective of what the remaining included users report. 

\begin{definition}[Restricted UIC (RUIC)] \label{def::ruic}
Given a TFM $\mathcal{T} = (\mathbf{x}^I,\mathbf{x}^C,\mathbf{p})$, we say that RUIC is satisfied if, $\forall i$ included in the block i.e., $\forall i \in I$ we have,
$$
u_i(\theta_i,b_i^\star=\theta_i,\mathbf{b}_{I\setminus i}) \geq u_i(\theta_i,b_i,\mathbf{b}_{I\setminus i}), \forall \theta_i, \mathbf{b}_{I\setminus i}
$$
where $\mathbf{b}_{I\setminus i}$ is the bids of users included in the block excluding user $i$.
\end{definition}




\paragraph{Other Properties.} Outside of these common TFM properties, we also define additional properties, namely allocative efficiency (AE) and weakly/strongly budget balance (WBB/SBB) next.  

\begin{definition}[Allocative Efficiency (AE)]
\label{def:ae}
We say that a TFM $\mathcal{T} = (\mathbf{x}^I,\mathbf{x}^C,\mathbf{p})$ satisfies AE, if given $\Theta$, the mechanism confirms the transactions which maximizes the overall valuation. That is, for any given $\Theta$ and every feasible allocation $\mathbf{x}^C$, we have:
$
\mathbf{x}^\star := \argmax_{\mathbf{x}^C} \sum_{i\in M} x_i^C\cdot \theta_i 
$
\end{definition}

\begin{definition}[Weakly/Strongly Budget Balance (WBB/SBB)]
We say that a TFM $\mathcal{T}=(\mathbf{x}^I,\mathbf{x}^C,\mathbf{p})$ satisfies WBB, if the total payment to the miner is non-negative, i.e., $\sum_{i \in I} p_i\geq 0$. When the equality holds, a TFM is strongly budget balanced (SBB), i.e., $\sum_{i\in I} p_i = 0$.
\end{definition}

\subsection{Groves' Redistribution Mechanism (RM)}
\label{subsec:prelim_rm}
Towards minimizing the user cost in a TFM, we employ \emph{Redistribution Mechanisms} (RMs)~\cite{LaffontM79}. In RM, the agents are charged VCG payments and the money is redistributed back to agents while ensuring UIC. The redistribution  is decided by constructing an appropriate rebate function, $g: \mathbf{b} \rightarrow \mathbb{R}$. We desire rebate functions that ensure maximum rebate (or, equivalently, minimize the transaction fees in TFM). To utilize an RM as a TFM, we require that it satisfies UIC/RUIC and (ex-post) IR for the users and the miner. An RM is IR for users when each user's overall payment, including the rebate, provides a non-negative utility. Likewise, we say that an RM is IR for the miner when the total rebate is less than the payment (transaction fees) received. We also want the RM to be \emph{anonymous}~\cite{Guo09}.


 \begin{definition}[Anonymity]\label{def:anonymity}
An RM satisfies anonymity if the rebate function is same for all the users, i.e., $\forall i,j\in [n]$ and $i\not=j$, $g_{i}(\cdot) = g_{j}(\cdot) = g(\cdot)$. 
 \end{definition}
 
Note that, an anonymous rebate function may still result in different redistribution payments to different users as the input to the function may be arbitrarily different.


\paragraph{Rebate Function.}
We aim to design an appropriate rebate function for an anonymous RM such that incentive properties from Section~\ref{subsec:Tfm-dp} hold. The rebate function must also redistribute most of the payments (VCG payments) as possible to minimize the user cost. We begin by providing the following characterization for designing UIC rebate functions.

\begin{theorem}[\cite{Gujar11}] \label{thm:order} In an RM, any deterministic, anonymous rebate function $g(\cdot)$ is UIC iff the rebate for user $i$ is defined as 
$r_{i} := g(b_{1},b_{2},\ldots,b_{i-1},b_{i+1},\ldots,b_{n}),~ \forall i \in [n],$
where $b_{1} \geq b_{2} \geq \ldots \geq b_{n}.$
\end{theorem}
The rebate function is UIC if the rebate for a user $i$ is independent of its own bid. In general, it could take any form. E.g., the linear rebate function is defined as,
 
\begin{definition}[Linear Rebate Function~\cite{Guo07}] \label{def:lr}
The rebates to a user $i$ follow a linear rebate function if the rebate is a linear combination of the bid vectors of all the remaining agents. That is, $r_{i} = c_0 + c_1 b_1 +\ldots + c_{i-1} b_{i-1}+ c_i b_{i+1} + \ldots + c_{n-1}b_{n-1}$ where $c_j\in\mathbb{R},~\forall j$. 
\end{definition}




The fraction of VCG payment redistributed, given a rebate function, depends on the input bids. Thus, we study the worst-case and average-case performance of the rebate functions.


\begin{definition}[Worst-case/Average Redistribution Index (RI)~\cite{Guo07}] \label{def:RI} 
The Worst-case or Average RI of an RM is defined as the worst-case or average-case fraction of VCG surplus that gets redistributed among the users, respectively. That is, given that $p(\mathbf{b}) = \sum_i p_i(\mathbf{b})$, the total VCG payment collected: 
\begin{equation}
    e_{\textsf{wc}} = \inf_{\mathbf{b}:p(\mathbf{b})\neq 0} \frac{\sum_{i} r_{i}}{p(\mathbf{b})} \mbox{~and~}  e_{\textsf{avg}} = \mathbb{E}_{\mathbf{b}:p(\mathbf{b})\neq 0}\left[ \frac{\sum_{i} r_{i}}{p(\mathbf{b})}\right]
\end{equation}

\end{definition}

E.g., \citet{Guo09} propose \emph{Worst-case Optimal (WCO)}, a mechanism that uniquely maximizes the worst-case RI (among all RMs that are deterministic, anonymous, and satisfy UIC, AE, and IR) when the items are homogeneous with unit demand~~\cite[Theorem 1]{Guo09}. In the blockchain, we assume the availability of $k$ confirmation slots, and any user equally values transaction confirmation at every slot. Moreover, each transaction only requires one slot for confirmation. Thus, this is a homogeneous setting with unit demand.


\section{\texttt{Ideal-TFRM}: Impossibility of Achieving Strictly Positive Redistribution Index\label{subsec_ideal_tfrm_imp}}

We now present a first attempt at implementing an RM for minimizing transaction fees using the second-price TFM  from Example~\ref{ex1}). In a second-price TFM, the transactions/bids are sorted in decreasing order. The top $k$ bids are confirmed with the $n=(k+1)^{th}$ transaction as the price-setting one. Each confirmed transaction pays $p = b_{k+1}$; with the net miner revenue as $k\cdot b_{k+1}$. To minimize the user fees, we must ``redistribute'' the collected surplus.

It may be preferable to only provide rebates to users whose transactions are confirmed (i.e., are among the top $k$ bids) since each such user pays $b_{k+1}$. If we also provide rebates to the remaining $n-k$ users, the remaining transactions in the mempool may prefer to overbid just enough also to get included in the block. By doing so, they can grab rebates for free\footnote{Assigning future costs to transactions not confirmed, e.g., as in \cite{chung2021foundations} (Section~\ref{sec::rw}), may help overcome such manipulation. We leave the analysis for future work.}. Thus, to achieve UIC, we must not provide rebates to unconfirmed transactions. With this motivation, we propose the following.

\smallskip
\noindent\underline{\texttt{Ideal-TFRM}}. The goal is to maximize the fraction of VCG payments redistributed to the users, denoted by $f$, while ensuring non-zero rebates only to confirmed transactions. Further, in \texttt{Ideal-TFRM}, we would like the rebate offered to each user to be less than the payment it makes. Eq.~\ref{eqn_ideal_tfrm} captures this optimization.
\begin{equation}\label{eqn_ideal_tfrm}
\begin{rcases}
    \max_{r_i, i\in \mathbf{x}} f
    \mbox{~\textbf{s.t.}~} &  \sum_{i\in C} r_i\geq f\cdot\sum_{i\in C} p_i \\
    & \mbox{~\textbf{and}~} p_i\geq r_i,~\forall i \in C \mbox{~\textbf{and}~}
    r_i=0,~\forall i \in I\setminus C
\end{rcases}
\end{equation}

Here, $r_i = g(b_{1},b_{2},\ldots,b_{i-1},b_{i+1},\ldots,b_{n})$ is the rebate (Definition~\ref{def:lr}) for each user $i\in M$ with $p_i$ as the VCG payment. The goal is to find an optimal $g(\cdot)$ such that $e_{\textsf{wc}}$ is maximized.

Unfortunately, we now show that for both the worst and average-case, \texttt{Ideal-TFRM} admits zero rebates for the users with confirmed transactions, i.e., $r_i=0,~\forall i \in\mathbf{x}^C$ while guaranteeing UIC.

\subsubsection*{Worst-case Rebate} Theorem~\ref{thm::ideal_tfrm_imp} formally shows the impossibility of simultaneously guaranteeing UIC and minimizing transaction fees in \texttt{Ideal-TFRM}.

\begin{theorem}[\texttt{Ideal-TFRM} Impossibility]\label{thm::ideal_tfrm_imp}
If $r^\star$ is an anonymous rebate function that satisfies Theorem \ref{thm:order},  no \texttt{Ideal-TFRM} can guarantee a non-zero redistribution index (RI) in the worst case.
\end{theorem}

\begin{proofsketch} Informally, if the rebate function is anonymous (Definition ~\ref{def:anonymity}) and the user with the highest valued transaction is confirmed and receives a positive rebate. Then, it is easy to show that there exists a different bid profile for which the last unconfirmed transaction must receive the same positive rebate. Therefore, the only way to ensure zero rebate for unconfirmed transactions is to also ensure zero rebate for every confirmed transaction. Thus, the worst-case RI is zero. See Appendix~\ref{app:thm_proof_ideal_tfrm} for the formal proof.
\end{proofsketch}





\subsubsection*{Average-case Rebate} Theorem~\ref{thm::ideal_tfrm_imp} shows that $e_{\textsf{wc}}=0$ in \texttt{Ideal-TFRM}, for a linear rebate function. We now aim to find a \textit{non-linear} rebate function that maximizes $e_{\textsf{avg}}$ in \texttt{Ideal-TFRM}. However, it is analytically intractable to characterize similar results to show the outcome of a rebate function that maximizes $e_{\textsf{avg}}$. As such, we simulate the optimization in Eq.~\ref{eqn_ideal_tfrm} as a Neural Network (NN), similar to \cite{manisha2018,tacchetti2019neural,dutting2019optimal}.

\paragraph{Architecture \& Setup.}  We consider a typical 3-layer feed-forward NN with bias, ReLU activation, and with AdamW optimizer. The input to our NN is the $n$-dimensional bid vector $\mathbf{b}_I$ sampled from a specific distribution. Each hidden layer comprises $2n$ neurons, with $n$ as the output layer's dimension. Given $\mathbf{b}_I$, the NN computes the payments and rebates to the confirmed and included transactions. 

\paragraph{Loss Function.} For optimization, our loss function is a weighted sum of the following three quantities: (i) average rebate to the $n$ bidders (denote as $r_\textsf{avg}$), (ii) feasibility, i.e., $\sum_i r_i \leq \sum_i p_i$ (denote as $r_\textsf{feas}$) and (iii) zero-rebate, i.e., $r_i=0,  \forall i \in I\setminus C $  (denote as $r_\textsf{zero}$. More concretely, for weights $\beta_1,\beta_2 \in (0,1)$  the loss function takes the form:
$
\textsf{Loss} = r_\textsf{avg} + \beta_1 \cdot r_\textsf{feas} + \beta_2 \cdot r_\textsf{zero}.$

\paragraph{Training Details.} We keep $n=10$ with number of confirmed transactions as $k=7$. For the optimizer, we choose a fixed learning rate $\eta = 5e-4$. The batch size is 1000, and we train for 50,000 epochs.

\paragraph{Results.} We observe that $e_{\textsf{avg}}\approx 0$ when transactions are sampled from $\mathcal{U}[0,1]$ and $\mathcal{N}(0,1)$. That is, the average case rebate to confirmed transactions is zero, even with non-linear rebate functions.

We conclude that it is impossible to design a TFRM with a linear rebate function that is UIC~(Theorem \ref{thm:order}) and offers a non-zero rebate to any agent. Our experiments also highlight that this may also be unlikely for non-linear rebate functions. Therefore, the next section introduces the general TFRM framework, where we focus on restricted UIC.

\section{Transaction Fee Redistribution Mechanism (TFRM)\label{sec::tfrm}}

As both $e_{\textsf{avg}}\approx e_{\textsf{wc}}=0$ for \texttt{Ideal-TFRM}, we must also provide rebates to users whose transactions are included but not confirmed. With this, we present the general TFRM framework in Figure~\ref{fig:TFRM_framework}.

\begin{figure}[t]
\small
\noindent\fbox{\parbox{0.97\columnwidth}{
\begin{enumerate}[leftmargin=*,itemsep=0.2em]
    \item \textsl{Inclusion Rule $(\mathbf{x}^I)$.} Select highest $n$ transactions from the mempool, $M$. W.l.o.g., assume that these $n$ transactions are ordered as $b_1\geq b_2\geq \ldots \geq b_n \implies x_i^I=1,~\forall i \in \{1,\ldots, n\}$.
    \item \textsl{Confirmation Rule $(\mathbf{x}^C)$.} Select highest $k$ bids from the $n$ included, $x_i^C=1,\forall i \in [k]$, where $k\le n-2$.
    \item \textsl{Payment Rule $(\mathbf{p})$.} Each confirmed user $i$ (i.e., $i\in C$) pays $p_i = b_{k+1} - r_i$. Each included but not confirmed user $j$ ($j\in I\setminus C$)  pays $p_j = - r_j$. 
    \item \textsl{Miner Revenue Rule.} The miner receives the net revenue of $\sum_{i=1}^n p_i$.
\end{enumerate}
}}
\caption{\label{fig:TFRM_framework} Transaction Fees Redistribution Mechanism (TFRM): General Framework.}
\end{figure}


In a TFRM, out of the $m$ outstanding transactions in the mempool, we include the $n$ highest bids in the block (denoted by the set $I$). Among the bids in the block, we confirm the $k$ highest bids (denoted by the set $C$) where $n\geq k+2$. The remaining bids (denoted by the set $P$) are included but not confirmed; we refer to them as included or price-setting transactions. That is, $|I| = n, |C| = k$ and $|P|=n-k$.  W.lo.g., we assume that $b_1 \ge b_2 \ge \ldots \ge b_n$. Hence, $C=\{b_1, \ldots b_k\}$ and $P = \{b_{k+1}, \ldots, b_n\}$. Each user $i$'s payment is computed based on the VCG payments and a rebate function, with $r_i$ as the rebate to user $i$. Since $k$ bids are confirmed, the VCG payment for the confirmed transactions is the $(k+1)^{th}$ highest bid, i.e., $b_{k+1}$.

\paragraph{TFRM: RUIC.}  While the rebate function satisfies Theorem \ref{thm:order}, note that TFRM is not UIC. E.g., users not part of the block may report $b>\theta$ to grab the additional rebate, as only confirmed bids pay the transaction fee. However, TFRM satisfies RUIC, i.e., it is UIC for agents included in the block to bid their true valuation. All included users, confirmed or not, are offered rebates implying that $k$ slots offered to $n$ users is equivalent to the allocation of $k$ resources among $n$ agents. Thus, by Theorem \ref{thm:order}, TFRM is RUIC.


We now show that in the presence of a strategic miner, the TFRM in Figure~\ref{fig:TFRM_framework} results in net zero rebates to confirmed users. 

\paragraph{TFRM: Effect of Strategic Miners.}
In general, strategic miners may introduce fake transactions to increase their revenue. In the second-price TFM itself (Example~\ref{ex1}), the miner may introduce a fake bid $\hat{b}_{k+1} = b_k$ to increase its revenue to $k\cdot b_k$ from the intended $k\cdot b_{k+1}$. 
Similarly, fake transactions can also affect the rebate offered by a TFRM. The miner's deviation may result in the following: (i) Fake bids affect the rebate of all agents, potentially reducing the rebate and (ii) As a fake bidder, the miner pays the rebate to itself.

Thus, designing TFRMs to minimize transaction fees may only work if made resilient to such strategic manipulations. Unlike RMs, TFRMs must quantify the rebate redistributed to the genuine users. Towards this, we define the following metric:

\begin{definition}[Resilient Redistribution Index (RRI)]
\label{def:rri}
    Given that the miner manipulates the bids $\mathbf{b}$ to $\hat{\mathbf{b}}$,  RRI is the fraction of the received payments that are redistributed in the worst-case to the actual users. Given $C$ confirmed and $P$ unconfirmed users, let $S \subseteq I$ be the subset of users that are not impersonated by the miner. Then:
    $\hat{e}_{\textsf{wc}} = \inf_{\hat{\mathbf{b}}, p(:,\hat{\mathbf{b}}) \neq 0} \frac{\sum_{i\in S} r_i}{p(\cdot;\hat{\mathbf{b}})}.$ 
\end{definition}

\paragraph{TFRM: Impossibility of Strictly Positive RRI.}
It is desirable to have a TFRM that is AE, RUIC, IR$_\mathbf{M}$ and IR$_u$ while ensuring strictly positive RRI in the worst-case, i.e., $\hat{e}_{\textsf{wc}} > 0$. Unfortunately, Theorem~\ref{thm::imp_RRI} proves that it is impossible to design such a TFRM with strategic miners. Appendix~\ref{app::thm_imp_rri} provides the formal proof.

\begin{theorem}[TFRM: RRI Impossibility]\label{thm::imp_RRI}
    Given a strategic miner, it is impossible to design a TFRM with a linear rebate function that is RUIC, AE, both IR$_u$ and IR$_\mathbf{M}$, and guarantees a strictly positive RRI, i.e., $\hat{e}_{\textsf{wc}}>0$.
\end{theorem}

\begin{proofsketch} A linear rebate that is RUIC and IR$_u$ and IR$_\mathbf{M}$ must depend only on bids $b_{k+2}, \ldots, b_{n}$. Changing these bids does not affect the payment received by the miner. Hence, the miner can replace these bids with fake bids such that the rebate is zero without any change in the payments.
\end{proofsketch}

From these results, we establish that preventing user manipulation entirely is not possible in the TFRM framework. Therefore, we focus on ensuring Restricted UIC (RUIC), which ensures that users of transactions included in the block will not misreport their values. Further, we know from Theorem \ref{thm::imp_RRI} that even with RUIC, any known RM that satisfies all the desirable properties can be easily manipulated by the miner. The theorem also shows that the manipulation will lead to strictly zero rebate. Thus, in the next section, we propose Robust TFRM (\rtfrm), which relaxes IR$_\mathbf{M}$, to ensure a positive rebate even with miner manipulation.


\section{\rtfrm: A TFRM Robust to Miner Manipulation\label{sec::tfrm-nonir}}

To ensure strictly positive RRI, we compromise on IR$_\mathbf{M}$, i.e., the utility of an honest miner may be negative. However, we ensure that when the miner is strategic, it can always guarantee itself a non-zero utility. We denote such a TFRM that is resilient to miner manipulation by Robust TFRM (\rtfrm). Designing \rtfrm\  involves constructing an appropriate rebate function. We focus on a linear rebate function that maximizes the worst-case redistribution index RRI while ensuring IR$_u$. We still want RUIC; hence we use the rebate function as given in Theorem \ref{thm:order}. The proofs of the results presented in this section are deferred to Appendix~\ref{app::section4}.

\paragraph{IR$_u$ Constraints.} Each included user must have a non-negative utility, i.e., $u_i \ge 0, \forall i\in I$. W.l.o.g., we assume that $b_1 \ge b_2 \ge \ldots b_n$ and IR for user $n$ is ensured when $r_n \ge 0$ as $u_n  = r_n$. In Claim~\ref{claim:ir-nonir-tfrm}, we show that $r_i \geq 0, \forall i$ if $r_n \ge 0$; hence \rtfrm\ will be IR$_u$.

\begin{claim}
\label{claim:ir-nonir-tfrm}
\rtfrm\ with $n$ included transactions and rebates $(r_1,\ldots,r_n)$ is user IR (IR$_u$) if $r_n \ge 0$.
\end{claim}

\paragraph{Approx-IR$_\mathbf{M}$ Constraints.} In any classic RM, to ensure IR$_\mathbf{M}$, there is an additional constraint to ensure that the total rebate is less than the VCG payments, i.e., $\sum_i r_i \leq k\cdot b_{k+1}$. We modify this constraint to ensure $\sum_i r_i \leq k\cdot b_k$. This change is because a strategic miner can manipulate the VCG auction to insert a fake bid $\hat{b}_{k+1} = b_k$. Figure~\ref{fig:tfrm-lp} describes the linear program to solve for such a rebate while maximizing the RRI fraction $f$.

\begin{figure}[t]
\noindent\fbox{\parbox{\columnwidth}{
\noindent\textbf{Maximize:} $f$ \quad \quad \textcolor{blue}{$\triangleright$\texttt{RRI}} \\
\noindent\textbf{Subject to:} For every $\mathbf{b}\geq 0,$
\begin{align*}
    r_n \ge 0 & \textcolor{blue}{\texttt{~$\triangleright$IR}_u}\\
   \sum_{i=1}^n r_i \leq k\cdot b_k  & \textcolor{blue}{\texttt{~$\triangleright$Approx-IR}_m}\\
 \noindent\sum_{i=1}^k r_i \ge f \cdot k\cdot b_{k+1}  &\textcolor{blue}{\texttt{~$\triangleright$Worst-case Fraction}} \\
\mbox{where~}  r_i = c_0+c_1b_1+\ldots+c_{i-1}b_{i-1}  
     & + c_i b_{i+1} + \ldots + c_{n-1}b_n
\end{align*}
}}
    \caption{\rtfrm: Linear Program for Rebate Function with Approx-IR$_m$}
    \label{fig:tfrm-lp}
\end{figure}

\begin{figure}
\noindent\fbox{\parbox{\columnwidth}{
\noindent\textbf{Maximize:} $f$ \\
\noindent\textbf{Subject To:} $\sum_{j=k}^i c_j \geq f, ~\forall i \in \{k, \ldots, n-1\}$
\begin{align*}
    (n-k)\cdot c_k \leq k &\mbox{~\textbf{and}~}   n \sum_{j=k}^{n-1}c_j \leq k \\
    n\sum_{j=k}^{k+i-1}c_j + (n-k-i)&\cdot c_{k+i} \leq  k,~\forall i\in [n-k-1]
\end{align*}
}}
    \caption{\rtfrm: Linear Program Independent of the bid vector $\mathbf{b}$ }
    \label{fig::nonir-tfrm}
\end{figure}

We now aim to write the linear program in Figure~\ref{fig:tfrm-lp}  such that it is only dependent on $n, k, c_i's$ and independent of the bid vector, $\mathbf{b}$. For this purpose, we now state the following claims.

\begin{claim}
    \label{claim:12k-nonir-tfrm}
    If $c_0, \ldots, c_{n-1}$ satisfy IR$_u$ and Approx-IR$_\mathbf{M}$, then $c_i = 0$ for $i = 0, \ldots, k-1$.
\end{claim}

\begin{claim}
    \label{claim:iru-nonir-tfrm}
    The IR$_u$ constraint $r_n \ge 0$ and the worst-case fraction constraint (refer Figure~\ref{fig:tfrm-lp}) is equivalent to having $\sum_{j=k}^i c_j \geq f,~\forall i \in \{k, \ldots, n-1\}$.
\end{claim}

\begin{claim}
    \label{claim:irmwc-nonir-tfrm}
    The Approx-IR$_\mathbf{M}$ constraint can be replaced by: 
    \begin{align*}
    (n-k)c_k \leq k\mbox{~\textbf{and}~} &
    n \cdot\sum_{j=k}^{n-1}c_j \leq k \mbox{~\textbf{and}~} \\
 n\sum_{j=k}^{k+i-1}c_j + (n-k-i) c_{k+i} \leq &  k,~ 
        i\in \{1, \ldots, n-k-1\}
    \end{align*}
\end{claim}

Using Claims \ref{claim:iru-nonir-tfrm} and \ref{claim:irmwc-nonir-tfrm}, we reformulate the linear program in Figure~\ref{fig:tfrm-lp}  so that it is independent of the bid vectors. Figure~\ref{fig::nonir-tfrm} presents this reformulated LP.

\paragraph{Optimal worst-case Redistribution Fraction.} We next provide the analytical solution to the linear program in Figure~\ref{fig::nonir-tfrm} and thereby also state the optimal worst-case fraction redistributed.

\begin{theorem}
    \label{thm:nonir-tfrm}
    For any $n$ and $k$ such that $n\ge k+2$, the \rtfrm\ mechanism is unique. The fraction redistributed to the top-k users in the worst-case is given by:
    $f^{*} = \frac{k}{n}.$
    In \rtfrm, the rebate function is characterized by the following:
    $c_k^{*} = \frac{k}{n}$ and $c_i^{*} = 0,~ \forall{i\neq k}.$
\end{theorem}

 Observe that the total redistribution to the users, when the miner is honest for \rtfrm, is given by 
 \begin{equation}
     \sum_{i=1}^n r_i = \frac{k}{n} \left[(k\cdot b_{k+1}) + (n-k)b_k\right]
     \label{eq:rebatenonirtfrm}
 \end{equation}
This value may exceed $k\cdot b_{k+1}$, thus violating IR$_\mathbf{M}$. However, it satisfies Approx-IR$_\mathbf{M}$ as $\sum_{i=1}^n r_i \leq k\cdot b_k$ (refer Figure~\ref{fig:tfrm-lp}).  \rtfrm\ is similar to the Bailey-Cavallo mechanism~\cite{Bailey97}. The primary difference is due to the Approx-IR$_\mathbf{M}$ constraint, which makes $c_k = k/n$ instead of $c_{k+1}$ as in Bailey-Cavallo. 


\subsection{\rtfrm: Analyzing Impact of Miner Manipulation on Rebate and Miner Revenue\label{sec::exp}}
With an honest miner, \rtfrm~ maximizes the worst-case redistribution index such that it is AE, RUIC, IR$_u$ and Approx-IR$_\mathbf{M}$. We now analyze the effect of miner manipulation on \rtfrm. Previously, we saw that it is impossible to ensure non-zero RRI (Theorem~\ref{thm::imp_RRI}), but with \rtfrm\, we show that RRI is strictly positive even with miner manipulation.

\paragraph{Reduction in Transaction Fees.} The rebate function for \rtfrm\ is characterized by the constants given in Theorem \ref{thm:nonir-tfrm}. With these, We now calculate RRI (Definition \ref{def:rri}), i.e, $\hat{e}_{\textsf{wc}}$, for \rtfrm. Theorem~\ref{thm:rri-nonir-tfrm} shows that irrespective of miner manipulation, $c^*_k=\frac{k}{n}$ fraction of payments will be returned.


\begin{theorem}
\label{thm:rri-nonir-tfrm}
  Consider $n$ included transactions with the set $C$ as confirmed transactions such that $|C|=k$ with the remaining $n-k$ as price-setting transactions. Irrespective of any miner manipulation, \rtfrm\ ensures strictly positive RRI or $ \hat{e}_{\textsf{wc}} =  c_k^* = \frac{k}{n}$.
\end{theorem}

From Theorem~\ref{thm:nonir-tfrm} and Theorem~\ref{thm:rri-nonir-tfrm}, we see that in \rtfrm, the fraction of payments redistributed to the top-k users, i.e., $k/n$, is the same for honest and strategic miner. This implies that \rtfrm\ is resilient to miner manipulation while being worst-case optimal.

\paragraph{Utility of Strategic Miner.} From Theorem \ref{thm:rri-nonir-tfrm}, we know that if a miner is strategic and impersonates the price-setting transactions, the miner will receive positive utility. The miner will preferably set the fake bid $\hat{b}_{k+1}$ close to $b_k$. Hence, the maximum utility to a miner that deviates by impersonating the price-setting bids is: $u_{\mathbf{M}} = \left(1-k/n\right) \cdot k\cdot b_{k}$. As we assume $n\geq k+2$, the miner's maximum utility is minimized for $k=n-2$. 

The fraction redistributed to the genuine users is still $\frac{k}{n}$ of the payments received even when the miner impersonates the confirmed transactions.  We illustrate this with an example next.

\begin{example}\upshape
It is possible for the miner to insert fake transactions with high enough bids such that they are confirmed. Consider $n=5$ and $k=3$ where $b_1 = b_2 = 100$, $b_3 = 10 $ and $b_4 = b_5 = 4$. If the miner is only impersonating the price setting transactions, then it puts $\hat{b}_4 = b_3$ and arbitrary $\hat{b_5}<b_3$,  then its overall utility is $\left ( 1 - \frac{k}{n} \right) k b_k = 12$. Whereas if the miner is given more flexibility to insert a fake transaction within the confirmed and unconfirmed bids, it receives more payments. For e.g., let $\hat{b}_1 = 200$ and $\hat{b}_3 = 100$ hence the ordered bids are $\hat{b}_1 \ge b_1 \ge b_2 \ge \hat{b}_3$. Therefore effectively, the first two transactions are confirmed and pay $100$ each.  Further, due to \rtfrm, it returns a rebate of $\frac{k}{n} 100$ to each of the two users, thus obtaining an overall utility of $\left( 1 - 3/5\right) 200 = 80$. 
\end{example}

\section{\rrtfrm: Robust and Rational TFRM}
\label{sec:rrtfrm}
\rtfrm\ compromises Miner IR to ensure positive RRI. We now introduce randomness in \rtfrm\ to obtain a mechanism that ensures positive utility to an honest miner, i.e., satisfies $IR_m$. Towards this, we propose Robust and Rational TFRM, \rrtfrm\ (Figure \ref{fig:RRTFRM_framework}). In \rrtfrm, the rebate is not guaranteed for every included transaction. Instead, an included transaction gets a rebate with probability $\alpha$, $\alpha \in [0, 1]$ where the rebate value is calculated using \rtfrm. Hence \rrtfrm\ reduces to \rtfrm\ when $\alpha=1$. On the other extreme, when $\alpha=0$, \rrtfrm\ reduces to a second price auction.



\paragraph{\rrtfrm: On-chain Randomness.} As stated, each transaction receives a rebate with a probability $\alpha$. Similar to other TFMs~\cite{chung2021foundations}, we employ trusted on-chain randomness for this randomization. Researchers have proposed such trusted randomized protocols using various cryptographic primitives~\cite{bhat2021randpiper,das2022spurt}. Significantly, the miner of the block cannot exert any influence on this randomization.

\begin{figure}[t]
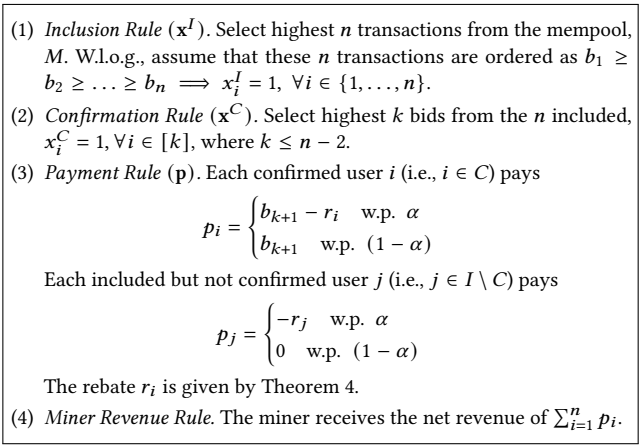

\small
\noindent\fbox{\parbox{0.97\columnwidth}{
\begin{enumerate}[leftmargin=*,itemsep=0.2em]
    \item \textsl{Inclusion Rule $(\mathbf{x}^I)$.} Select highest $n$ transactions from the mempool, $M$. W.l.o.g., assume that these $n$ transactions are ordered as $b_1\geq b_2\geq \ldots \geq b_n \implies x_i^I=1,~\forall i \in \{1,\ldots, n\}$.
    \item \textsl{Confirmation Rule $(\mathbf{x}^C)$.} Select highest $k$ bids from the $n$ included, $x_i^C=1,\forall i \in [k]$, where $k\le n-2$.
    \item \textsl{Payment Rule $(\mathbf{p})$.} Each confirmed user $i$ (i.e., $i \in C $) pays 
    \begin{equation*}
        p_i = 
        \begin{cases}
            b_{k+1} - r_i \quad \mbox{w.p.}\ \  \alpha \\
            b_{k+1} \quad \mbox{w.p.}\  \ (1-\alpha)
        \end{cases}
    \end{equation*}
     Each included but not confirmed user $j$ (i.e., $j\in I\setminus C$)  pays 
    \begin{equation*}
        p_j = 
        \begin{cases}
             - r_j \quad \mbox{w.p.}\ \  \alpha \\
            0 \quad \mbox{w.p.}\  \ (1-\alpha)
        \end{cases}
    \end{equation*}
    The rebate $r_i$ is given by Theorem \ref{thm:nonir-tfrm}.
    \item \textsl{Miner Revenue Rule.} The miner receives the net revenue of $\sum_{i=1}^n p_i$.
\end{enumerate}
}}
\caption{\label{fig:RRTFRM_framework} \rrtfrm: Robust and Rational Transaction Fees Redistribution Mechanism.}
\end{figure}

\paragraph{\rrtfrm: Incentive and RRI Guarantees.} Theorem \ref{thm:rrtfrm} proves that \rrtfrm~ mimics the incentive guarantees of \rtfrm. Moreover, for an appropriate $\alpha$, \rrtfrm\ is also IR$_\mathbf{M}$ in expectation. The proofs for the results in this section are given in Appendix \ref{app::sec7}.

\begin{theorem}
\label{thm:rrtfrm}
    For any $n$ and $k$ such that $n\geq k+2$ and any bid profile $\mathbf{b} = (b_1, \ldots, b_n)$, and probability $\alpha \in (0,1)$ \rrtfrm\ has an expected redistribution fraction (expectation over $\alpha$) $f^* = \alpha \cdot \frac{k}{n}$. Further it satisfies $AE$, $RUIC$, $IR_u$, and is $IR_m$ when $\alpha \leq \overline{\alpha}=\frac{n}{k+ (n-k)b_k/b_{k+1}} $. 
\end{theorem}

\begin{proofsketch}
We divide the proof into two parts.
\begin{itemize}[leftmargin=*]
    \item Firstly, notice that \rrtfrm\ is AE since the top $k$ bids are confirmed and IR$_u$, since payments are governed by VCG and all the rebates, are positive. Although in \rrtfrm\, every user is not deterministically given rebates, so it is RUIC in expectation, where the expectation is over $\alpha$. This is because the rebate a user receives is independent of its own bid.

   \item Next, we now formally show that \rrtfrm\ is IR$_\mathbf{M}$ in expectation. The payment obtained by an honest miner is given by $kb_{k+1}$. The expected rebate paid by the honest miner is given by $\alpha \left[k \frac{k}{n} b_{k+1} + (n-k) \frac{k}{n} b_k\right] $. If $\alpha = 1$, the refund is equal to the refund in \rtfrm\ given by Equation \ref{eq:rebatenonirtfrm}. In order to ensure $IR_m$ for the honest miner, the following must be true,

    \begin{align*}
        kb_{k+1} &\geq \alpha \left[k \frac{k}{n} b_{k+1} + (n-k) \frac{k}{n} b_k\right] \implies
        \alpha &\leq \frac{nb_{k+1}}{ k b_{k+1} + (n-k)b_k}
    \end{align*}
Therefore, \rrtfrm\ satisfies IR$_\mathbf{M}$ when $\alpha =  \frac{nb_{k+1}}{ k b_{k+1} + (n-k)b_k} $ and total rebate is $k b_{k+1}$. 
\end{itemize}
This completes the proof of the theorem.
\end{proofsketch}

We make the following observations based on Theorem~\ref{thm:rrtfrm}.

\noindent\textbf{$b_{k+1}\to b_{k} \implies$ \rtfrm $\iff$ \rrtfrm. } From Theorem \ref{thm:rrtfrm}, an honest miner obtains non-negative utility when, 
\begin{equation}
    \alpha \leq \frac{n}{k + (n-k) b_f}
    \label{eq:alpha}
\end{equation}
where,  $ b_f = \frac{b_{k}}{b_{k+1}}$ is the bid ratio. W.l.o.g as the bids are ordered (Figure \ref{fig:RRTFRM_framework}) $b_f \geq 1$. When $b_f = 1$, i.e., $b_k = b_{k+1}$ we have $\alpha = 1$ and \rtfrm $\iff$ \rrtfrm. This implies that every user receives a rebate, and the miner IR is not violated. 

This may seem to contradict \rtfrm\ not being IR$_\mathbf{M}$; however, we can think of $b_{k+1} \to b_k$ as one of the deviations of the strategic miner. And \rtfrm\ is IR$_\mathbf{M}$ when the miner is strategic. Furthermore, as $b_f$ increases, the upper bound on $\alpha$ becomes smaller. To still guarantee IR$_\mathbf{M}$, the overall rebate (i.e., $\alpha\cdot \frac{k}{n}$) decreases. Observe that, in the worst case, trivially, $\alpha = 0$. Generally, the bound on $\alpha$ will depend on the bid distribution. Sampling the bids uniformly $\mathbf{b} \sim Unif[0,1]$, the expected value of $\alpha$ (using order statistics) is, $ \mathbb{E}_{\mathbf{b}}[\alpha] = n/(n+1)$ (i.e., $\alpha$ tends to 1 for large values of $n$)

\paragraph{\rrtfrm: Analyzing Miner Manipulation.} Like \rtfrm, \rrtfrm\ also ensures strictly positive RRI even with miner manipulation as formally stated in Theorem \ref{thm:rri-rr-tfrm}.

\begin{theorem}
\label{thm:rri-rr-tfrm}
  Consider $n$ included transactions with the set $C$ as confirmed transactions such that $|C|=k$ with the remaining $n-k$ as price-setting transactions. Irrespective of any miner manipulation \rrtfrm\ ensures: expected RRI or $ \mathbf{E}_{\alpha}[\hat{e}_{\textsf{wc}}] =  \alpha\cdot \frac{k}{n}$.
\end{theorem}
\begin{proofsketch}
    The result follows from the similar result in Theorem \ref{thm:rri-nonir-tfrm} for \rtfrm\ and the fact that the miner has no control over the randomization in \rrtfrm.
\end{proofsketch}

We observe that irrespective of the miner manipulation; the users receive back the fraction $\alpha\cdot \frac{k}{n}$ of the payment made on expectation. Based on what fake transactions the miner inserts, the payments change, but the refund fraction remains the same as for the case when the miner is honest.


\paragraph{Discussion on TFRMs.} Given that the current TFM literature assumes that users and miners are myopic, we believe that redistributing the surplus is more effective in reducing the net payments paid by the users. Another desirable property in TFMs is that of \emph{predictable} transaction fees, i.e., reducing the volatility of the fees paid by the users. For instance, EIP-1559~\cite{buterin2019eip} uses a deterministic function based on the previous block consumption to calculate a network-determined minimum threshold fee (called \emph{base} fee) paid by each user. The users can also pay a \emph{priority} fee over the base fee to incentivize the miners to include their transactions. The base fee aims to reduce the fee volatility and aims to arrive at a market clearing price. Crucially, this minimum threshold fee is \emph{burned} -- transferred to an unspendable address -- implying that the miner does not receive this fee as revenue. We see that burning some fraction of the fee is necessary to guarantee such properties. In such mechanisms, the priority fee calculated post-burning can be further reduced by employing TFRMs. 

\section{Conclusion\label{sec::con}}
In this paper, we argued the importance of minimizing user costs in a TFM. Our key idea is to employ a redistribution mechanism-based approach for determining the transaction fees, which we call the Transaction Fee Redistribution Mechanism (TFRM). Due to strategic miner manipulation, we first show that guaranteeing a strictly positive rebate in a TFRM and other desirable properties is impossible. Hence, we propose \rtfrm, which ensures strictly positive rebates even in the worst case but compromises on miner's IR. However, we show that in \rtfrm, a strategic miner will never incur negative utility while still guaranteeing strictly positive rebates to the users. We also propose \rrtfrm\, which uses blockchain's inherent randomness to guarantee a strictly positive rebate to the users while also respecting the miner's IR.

\noindent\textbf{Future Work.} Future directions can explore TFRMs with randomized rebate functions, which may likely satisfy stronger notions of IC and IR. Another approach may be to explore non-linear rebate functions, which may provide a better redistribution index on average. In addition, unlike this work, future work can also explore transactions with varying sizes. 



\bibliographystyle{ACM-Reference-Format} 
\bibliography{ref}


\begin{thebibliography}{40}


\ifx \showCODEN    \undefined \def \showCODEN     #1{\unskip}     \fi
\ifx \showDOI      \undefined \def \showDOI       #1{#1}\fi
\ifx \showISBNx    \undefined \def \showISBNx     #1{\unskip}     \fi
\ifx \showISBNxiii \undefined \def \showISBNxiii  #1{\unskip}     \fi
\ifx \showISSN     \undefined \def \showISSN      #1{\unskip}     \fi
\ifx \showLCCN     \undefined \def \showLCCN      #1{\unskip}     \fi
\ifx \shownote     \undefined \def \shownote      #1{#1}          \fi
\ifx \showarticletitle \undefined \def \showarticletitle #1{#1}   \fi
\ifx \showURL      \undefined \def \showURL       {\relax}        \fi
\providecommand\bibfield[2]{#2}
\providecommand\bibinfo[2]{#2}
\providecommand\natexlab[1]{#1}
\providecommand\showeprint[2][]{arXiv:#2}

\bibitem[\protect\citeauthoryear{Bailey}{Bailey}{1997}]%
        {Bailey97}
\bibfield{author}{\bibinfo{person}{Martin~J Bailey}.} \bibinfo{year}{1997}\natexlab{}.
\newblock \showarticletitle{The Demand Revealing Process: To Distribute the Surplus}.
\newblock \bibinfo{journal}{\emph{Public Choice}} \bibinfo{volume}{91}, \bibinfo{number}{2} (\bibinfo{year}{1997}), \bibinfo{pages}{107--26}.
\newblock
\urldef\tempurl%
\url{https://EconPapers.repec.org/RePEc:kap:pubcho:v:91:y:1997:i:2:p:107-26}
\showURL{%
\tempurl}


\bibitem[\protect\citeauthoryear{Basu, Easley, O'Hara, and Sirer}{Basu et~al\mbox{.}}{2019}]%
        {basu2019stablefees}
\bibfield{author}{\bibinfo{person}{Soumya Basu}, \bibinfo{person}{David Easley}, \bibinfo{person}{Maureen O'Hara}, {and} \bibinfo{person}{Emin Sirer}.} \bibinfo{year}{2019}\natexlab{}.
\newblock \showarticletitle{StableFees: A Predictable Fee Market for Cryptocurrencie}.
\newblock \bibinfo{journal}{\emph{Available at SSRN 3318327}} (\bibinfo{year}{2019}).
\newblock


\bibitem[\protect\citeauthoryear{Bhat, Shrestha, Luo, Kate, and Nayak}{Bhat et~al\mbox{.}}{2021}]%
        {bhat2021randpiper}
\bibfield{author}{\bibinfo{person}{Adithya Bhat}, \bibinfo{person}{Nibesh Shrestha}, \bibinfo{person}{Zhongtang Luo}, \bibinfo{person}{Aniket Kate}, {and} \bibinfo{person}{Kartik Nayak}.} \bibinfo{year}{2021}\natexlab{}.
\newblock \showarticletitle{Randpiper--reconfiguration-friendly random beacons with quadratic communication}. In \bibinfo{booktitle}{\emph{ACM CCS}}. \bibinfo{pages}{3502--3524}.
\newblock


\bibitem[\protect\citeauthoryear{Buterin et~al\mbox{.}}{Buterin et~al\mbox{.}}{2014}]%
        {buterin2014next}
\bibfield{author}{\bibinfo{person}{Vitalik Buterin} {et~al\mbox{.}}} \bibinfo{year}{2014}\natexlab{}.
\newblock \showarticletitle{A next-generation smart contract and decentralized application platform}.
\newblock \bibinfo{journal}{\emph{White Paper}} \bibinfo{volume}{3}, \bibinfo{number}{37} (\bibinfo{year}{2014}), \bibinfo{pages}{2--1}.
\newblock


\bibitem[\protect\citeauthoryear{Buterin, Conner, Dudley, Slipper, Norden, and Bakhta}{Buterin et~al\mbox{.}}{2019}]%
        {buterin2019eip}
\bibfield{author}{\bibinfo{person}{Vitalik Buterin}, \bibinfo{person}{Eric Conner}, \bibinfo{person}{Rick Dudley}, \bibinfo{person}{Matthew Slipper}, \bibinfo{person}{Ian Norden}, {and} \bibinfo{person}{Abdelhamid Bakhta}.} \bibinfo{year}{2019}\natexlab{}.
\newblock \bibinfo{title}{EIP-1559: Fee market change for ETH 1.0 chain}.
\newblock \bibinfo{howpublished}{\url{eips.ethereum.org/EIPS/eip-1559}}.
\newblock


\bibitem[\protect\citeauthoryear{Cavallo}{Cavallo}{2006}]%
        {cavalloRedis06}
\bibfield{author}{\bibinfo{person}{Ruggiero Cavallo}.} \bibinfo{year}{2006}\natexlab{}.
\newblock \showarticletitle{Optimal Decision-Making With Minimal Waste: Strategyproof Redistribution of VCG Payments}. In \bibinfo{booktitle}{\emph{Proc. of the 5th Int. Joint Conf. on Autonomous Agents and Multi Agent Systems (AAMAS{\textquoteright}06)}}. \bibinfo{address}{Hakodate, Japan}, \bibinfo{pages}{882--889}.
\newblock
\urldef\tempurl%
\url{http://econcs.seas.harvard.edu/files/econcs/files/cavallo-redis.pdf}
\showURL{%
\tempurl}


\bibitem[\protect\citeauthoryear{Cavallo}{Cavallo}{2008}]%
        {cavallo08}
\bibfield{author}{\bibinfo{person}{Ruggiero Cavallo}.} \bibinfo{year}{2008}\natexlab{}.
\newblock \showarticletitle{Efficiency and Redistribution in Dynamic Mechanism Design}. In \bibinfo{booktitle}{\emph{Proc. 9th ACM Conf. on Electronic Commerce (EC{\textquoteright}08)}}. \bibinfo{address}{Chicago, IL}, \bibinfo{pages}{220--229}.
\newblock
\urldef\tempurl%
\url{http://econcs.seas.harvard.edu/files/econcs/files/cavallo-ec08.pdf}
\showURL{%
\tempurl}


\bibitem[\protect\citeauthoryear{Chung and Shi}{Chung and Shi}{2023}]%
        {chung2021foundations}
\bibfield{author}{\bibinfo{person}{Hao Chung} {and} \bibinfo{person}{Elaine Shi}.} \bibinfo{year}{2023}\natexlab{}.
\newblock \showarticletitle{Foundations of transaction fee mechanism design}. In \bibinfo{booktitle}{\emph{ACM-SIAM SODA}}. \bibinfo{pages}{3856--3899}.
\newblock


\bibitem[\protect\citeauthoryear{Clarke}{Clarke}{1971}]%
        {Clarke71}
\bibfield{author}{\bibinfo{person}{Edward Clarke}.} \bibinfo{year}{1971}\natexlab{}.
\newblock \showarticletitle{Multipart pricing of public goods}.
\newblock \bibinfo{journal}{\emph{Public Choice}} \bibinfo{volume}{11}, \bibinfo{number}{1} (\bibinfo{year}{1971}), \bibinfo{pages}{17--33}.
\newblock


\bibitem[\protect\citeauthoryear{Das, Krishnan, Isaac, and Ren}{Das et~al\mbox{.}}{2022}]%
        {das2022spurt}
\bibfield{author}{\bibinfo{person}{Sourav Das}, \bibinfo{person}{Vinith Krishnan}, \bibinfo{person}{Irene~Miriam Isaac}, {and} \bibinfo{person}{Ling Ren}.} \bibinfo{year}{2022}\natexlab{}.
\newblock \showarticletitle{Spurt: Scalable distributed randomness beacon with transparent setup}. In \bibinfo{booktitle}{\emph{IEEE S\&P}}.
\newblock


\bibitem[\protect\citeauthoryear{de~Clippel, Naroditskiy, Polukarov, Greenwald, and Jennings}{de~Clippel et~al\mbox{.}}{2014}]%
        {Clippel14}
\bibfield{author}{\bibinfo{person}{Geoffroy de Clippel}, \bibinfo{person}{Victor Naroditskiy}, \bibinfo{person}{Maria Polukarov}, \bibinfo{person}{Amy Greenwald}, {and} \bibinfo{person}{Nicholas~R. Jennings}.} \bibinfo{year}{2014}\natexlab{}.
\newblock \showarticletitle{Destroy to save}.
\newblock \bibinfo{journal}{\emph{Games and Economic Behavior}} \bibinfo{volume}{86}, \bibinfo{number}{C} (\bibinfo{year}{2014}), \bibinfo{pages}{392--404}.
\newblock


\bibitem[\protect\citeauthoryear{D{\"u}tting, Feng, Narasimhan, Parkes, and Ravindranath}{D{\"u}tting et~al\mbox{.}}{2019}]%
        {dutting2019optimal}
\bibfield{author}{\bibinfo{person}{Paul D{\"u}tting}, \bibinfo{person}{Zhe Feng}, \bibinfo{person}{Harikrishna Narasimhan}, \bibinfo{person}{David Parkes}, {and} \bibinfo{person}{Sai~Srivatsa Ravindranath}.} \bibinfo{year}{2019}\natexlab{}.
\newblock \showarticletitle{Optimal auctions through deep learning}. In \bibinfo{booktitle}{\emph{ICML}}. PMLR, \bibinfo{pages}{1706--1715}.
\newblock


\bibitem[\protect\citeauthoryear{Faltings}{Faltings}{2005}]%
        {Faltings05}
\bibfield{author}{\bibinfo{person}{Boi Faltings}.} \bibinfo{year}{2005}\natexlab{}.
\newblock \showarticletitle{A Budget-balanced, Incentive-compatible Scheme for Social Choice}. In \bibinfo{booktitle}{\emph{Agent-Mediated Electronic Commerce VI. Theories for and Engineering of Distributed Mechanisms and Systems: AAMAS 2004 Workshop, AMEC 2004}}. \bibinfo{pages}{30--43}.
\newblock


\bibitem[\protect\citeauthoryear{Ferreira, Moroz, Parkes, and Stern}{Ferreira et~al\mbox{.}}{2021}]%
        {parkes21}
\bibfield{author}{\bibinfo{person}{Matheus V.~X. Ferreira}, \bibinfo{person}{Daniel~J. Moroz}, \bibinfo{person}{David~C. Parkes}, {and} \bibinfo{person}{Mitchell Stern}.} \bibinfo{year}{2021}\natexlab{}.
\newblock \showarticletitle{Dynamic posted-price mechanisms for the blockchain transaction-fee market}. In \bibinfo{booktitle}{\emph{ACM Conference on Advances in Financial Technologies (AFT)}}. \bibinfo{pages}{86--99}.
\newblock


\bibitem[\protect\citeauthoryear{Groves}{Groves}{1973}]%
        {Groves73}
\bibfield{author}{\bibinfo{person}{Theodore Groves}.} \bibinfo{year}{1973}\natexlab{}.
\newblock \showarticletitle{Incentives in Teams}.
\newblock \bibinfo{journal}{\emph{Econometrica}} \bibinfo{volume}{41}, \bibinfo{number}{4} (\bibinfo{year}{1973}), \bibinfo{pages}{617--31}.
\newblock


\bibitem[\protect\citeauthoryear{Gujar and Narahari}{Gujar and Narahari}{2011}]%
        {Gujar11}
\bibfield{author}{\bibinfo{person}{Sujit Gujar} {and} \bibinfo{person}{Y. Narahari}.} \bibinfo{year}{2011}\natexlab{}.
\newblock \showarticletitle{Redistribution Mechanisms for Assignment of Heterogeneous Objects}.
\newblock \bibinfo{journal}{\emph{J. Artif. Int. Res.}} \bibinfo{volume}{41}, \bibinfo{number}{2} (\bibinfo{date}{May} \bibinfo{year}{2011}), 24.
\newblock


\bibitem[\protect\citeauthoryear{Guo}{Guo}{2012}]%
        {Guo12}
\bibfield{author}{\bibinfo{person}{Mingyu Guo}.} \bibinfo{year}{2012}\natexlab{}.
\newblock \showarticletitle{Worst-case Optimal Redistribution of VCG Payments in Heterogeneous-item Auctions with Unit Demand}. In \bibinfo{booktitle}{\emph{AAMAS}}. \bibinfo{pages}{745--752}.
\newblock


\bibitem[\protect\citeauthoryear{Guo and Conitzer}{Guo and Conitzer}{2007}]%
        {Guo07}
\bibfield{author}{\bibinfo{person}{Mingyu Guo} {and} \bibinfo{person}{Vincent Conitzer}.} \bibinfo{year}{2007}\natexlab{}.
\newblock \showarticletitle{Worst-case Optimal Redistribution of VCG Payments}. In \bibinfo{booktitle}{\emph{ACM EC}}. \bibinfo{pages}{30--39}.
\newblock


\bibitem[\protect\citeauthoryear{Guo and Conitzer}{Guo and Conitzer}{2008}]%
        {Guo08}
\bibfield{author}{\bibinfo{person}{Mingyu Guo} {and} \bibinfo{person}{Vincent Conitzer}.} \bibinfo{year}{2008}\natexlab{}.
\newblock \showarticletitle{Better Redistribution with Inefficient Allocation in Multi-unit Auctions with Unit Demand}. In \bibinfo{booktitle}{\emph{ACM EC}}. \bibinfo{pages}{210--219}.
\newblock


\bibitem[\protect\citeauthoryear{Guo and Conitzer}{Guo and Conitzer}{2009}]%
        {Guo09}
\bibfield{author}{\bibinfo{person}{Mingyu Guo} {and} \bibinfo{person}{Vincent Conitzer}.} \bibinfo{year}{2009}\natexlab{}.
\newblock \showarticletitle{Worst-case optimal redistribution of VCG payments in multi-unit auctions}.
\newblock \bibinfo{journal}{\emph{Games and Economic Behavior}} \bibinfo{volume}{67}, \bibinfo{number}{1} (\bibinfo{year}{2009}), \bibinfo{pages}{69--98}.
\newblock


\bibitem[\protect\citeauthoryear{Hartline and Roughgarden}{Hartline and Roughgarden}{2008}]%
        {Hartline08}
\bibfield{author}{\bibinfo{person}{Jason~D. Hartline} {and} \bibinfo{person}{Tim Roughgarden}.} \bibinfo{year}{2008}\natexlab{}.
\newblock \showarticletitle{Optimal Mechanism Design and Money Burning}. In \bibinfo{booktitle}{\emph{Proceedings of the Fortieth Annual ACM Symposium on Theory of Computing}} (Victoria, British Columbia, Canada) \emph{(\bibinfo{series}{STOC '08})}. \bibinfo{publisher}{ACM}, \bibinfo{address}{New York, NY, USA}, \bibinfo{pages}{75--84}.
\newblock
\showISBNx{978-1-60558-047-0}
\urldef\tempurl%
\url{https://doi.org/10.1145/1374376.1374390}
\showDOI{\tempurl}


\bibitem[\protect\citeauthoryear{Kelkar, Zhang, Goldfeder, and Juels}{Kelkar et~al\mbox{.}}{2020}]%
        {kelkar2020order}
\bibfield{author}{\bibinfo{person}{Mahimna Kelkar}, \bibinfo{person}{Fan Zhang}, \bibinfo{person}{Steven Goldfeder}, {and} \bibinfo{person}{Ari Juels}.} \bibinfo{year}{2020}\natexlab{}.
\newblock \showarticletitle{Order-fairness for byzantine consensus}. In \bibinfo{booktitle}{\emph{Annual International Cryptology Conference (CRYPTO)}}. \bibinfo{pages}{451--480}.
\newblock


\bibitem[\protect\citeauthoryear{Kursawe}{Kursawe}{2020}]%
        {kursawe2020wendy}
\bibfield{author}{\bibinfo{person}{Klaus Kursawe}.} \bibinfo{year}{2020}\natexlab{}.
\newblock \showarticletitle{Wendy, the good little fairness widget: Achieving order fairness for blockchains}. In \bibinfo{booktitle}{\emph{ACM Conference on Advances in Financial Technologies (AFT)}}. \bibinfo{pages}{25--36}.
\newblock


\bibitem[\protect\citeauthoryear{Manisha, Jawahar, and Gujar}{Manisha et~al\mbox{.}}{2018}]%
        {manisha2018}
\bibfield{author}{\bibinfo{person}{Padala Manisha}, \bibinfo{person}{CV Jawahar}, {and} \bibinfo{person}{Sujit Gujar}.} \bibinfo{year}{2018}\natexlab{}.
\newblock \showarticletitle{Learning Optimal Redistribution Mechanisms Through Neural Networks}. In \bibinfo{booktitle}{\emph{Proceedings of the 17th International Conference on Autonomous Agents and MultiAgent Systems}}. \bibinfo{pages}{345--353}.
\newblock


\bibitem[\protect\citeauthoryear{Maskin, Laffont, and Laffont}{Maskin et~al\mbox{.}}{1979}]%
        {LaffontM79}
\bibfield{author}{\bibinfo{person}{Eric Maskin}, \bibinfo{person}{J.~J. Laffont}, {and} \bibinfo{person}{J.~J. Laffont}.} \bibinfo{year}{1979}\natexlab{}.
\newblock \bibinfo{booktitle}{\emph{A Differential Approach to Expected Utility Maximizing Mechanisms}}.
\newblock \bibinfo{publisher}{North Holland}, \bibinfo{address}{289-308}.
\newblock


\bibitem[\protect\citeauthoryear{Messias, Alzayat, Chandrasekaran, and Gummadi}{Messias et~al\mbox{.}}{2020}]%
        {messias2020blockchain}
\bibfield{author}{\bibinfo{person}{Johnnatan Messias}, \bibinfo{person}{Mohamed Alzayat}, \bibinfo{person}{Balakrishnan Chandrasekaran}, {and} \bibinfo{person}{Krishna~P Gummadi}.} \bibinfo{year}{2020}\natexlab{}.
\newblock \showarticletitle{On Blockchain Commit Times: An analysis of how miners choose Bitcoin transactions}. In \bibinfo{booktitle}{\emph{The Second International Workshop on Smart Data for Blockchain and Distributed Ledger (SDBD2020)}}.
\newblock


\bibitem[\protect\citeauthoryear{Moulin}{Moulin}{2009}]%
        {Moulin09}
\bibfield{author}{\bibinfo{person}{Herve Moulin}.} \bibinfo{year}{2009}\natexlab{}.
\newblock \showarticletitle{Almost budget-balanced VCG mechanisms to assign multiple objects}.
\newblock \bibinfo{journal}{\emph{Journal of Economic Theory}} \bibinfo{volume}{144}, \bibinfo{number}{1} (\bibinfo{year}{2009}), \bibinfo{pages}{96--119}.
\newblock


\bibitem[\protect\citeauthoryear{Myerson and Satterthwaite}{Myerson and Satterthwaite}{1983}]%
        {myerson1983efficient}
\bibfield{author}{\bibinfo{person}{Roger~B Myerson} {and} \bibinfo{person}{Mark~A Satterthwaite}.} \bibinfo{year}{1983}\natexlab{}.
\newblock \showarticletitle{Efficient mechanisms for bilateral trading}.
\newblock \bibinfo{journal}{\emph{Journal of economic theory}} \bibinfo{volume}{29}, \bibinfo{number}{2} (\bibinfo{year}{1983}), \bibinfo{pages}{265--281}.
\newblock


\bibitem[\protect\citeauthoryear{Nakamoto}{Nakamoto}{2008}]%
        {nakamoto2008bitcoin}
\bibfield{author}{\bibinfo{person}{Satoshi Nakamoto}.} \bibinfo{year}{2008}\natexlab{}.
\newblock \showarticletitle{Bitcoin: A peer-to-peer electronic cash system}.
\newblock \bibinfo{journal}{\emph{Decentralized Business Review}} (\bibinfo{year}{2008}), \bibinfo{pages}{21260}.
\newblock


\bibitem[\protect\citeauthoryear{Orda and Rottenstreich}{Orda and Rottenstreich}{2021}]%
        {orda2021enforcing}
\bibfield{author}{\bibinfo{person}{Ariel Orda} {and} \bibinfo{person}{Ori Rottenstreich}.} \bibinfo{year}{2021}\natexlab{}.
\newblock \showarticletitle{Enforcing fairness in blockchain transaction ordering}.
\newblock \bibinfo{journal}{\emph{Peer-to-peer Networking and Applications}} \bibinfo{volume}{14}, \bibinfo{number}{6} (\bibinfo{year}{2021}), \bibinfo{pages}{3660--3673}.
\newblock


\bibitem[\protect\citeauthoryear{Parkes, Kalagnanam, and Eso}{Parkes et~al\mbox{.}}{2001}]%
        {parkes01}
\bibfield{author}{\bibinfo{person}{David~C. Parkes}, \bibinfo{person}{Jayant~R. Kalagnanam}, {and} \bibinfo{person}{Marta Eso}.} \bibinfo{year}{2001}\natexlab{}.
\newblock \showarticletitle{Achieving Budget-Balance with Vickrey-Based Payment Schemes in Exchanges}. In \bibinfo{booktitle}{\emph{Proc. 17th International Joint Conference on Artificial Intelligence (IJCAI{\textquoteright}01)}}. \bibinfo{pages}{1161{\textendash}1168}.
\newblock
\urldef\tempurl%
\url{http://econcs.seas.harvard.edu/files/econcs/files/combexch01.pdf}
\showURL{%
\tempurl}


\bibitem[\protect\citeauthoryear{Porter and Shoham}{Porter and Shoham}{2003}]%
        {porter2003cheating}
\bibfield{author}{\bibinfo{person}{Ryan Porter} {and} \bibinfo{person}{Yoav Shoham}.} \bibinfo{year}{2003}\natexlab{}.
\newblock \showarticletitle{On cheating in sealed-bid auctions}. In \bibinfo{booktitle}{\emph{Proceedings of the 4th ACM conference on Electronic commerce}}. \bibinfo{pages}{76--84}.
\newblock


\bibitem[\protect\citeauthoryear{Roughgarden}{Roughgarden}{2021}]%
        {roughgarden21}
\bibfield{author}{\bibinfo{person}{Tim Roughgarden}.} \bibinfo{year}{2021}\natexlab{}.
\newblock \showarticletitle{Transaction Fee Mechanism Design}. In \bibinfo{booktitle}{\emph{ACM Conference on Economics and Computation (ACM EC)}}. \bibinfo{pages}{792}.
\newblock


\bibitem[\protect\citeauthoryear{Siddiqui, Vanahalli, and Gujar}{Siddiqui et~al\mbox{.}}{2020}]%
        {bitcoinf}
\bibfield{author}{\bibinfo{person}{Shoeb Siddiqui}, \bibinfo{person}{Ganesh Vanahalli}, {and} \bibinfo{person}{Sujit Gujar}.} \bibinfo{year}{2020}\natexlab{}.
\newblock \showarticletitle{BitcoinF: Achieving Fairness For Bitcoin In Transaction Fee Only Model}. In \bibinfo{booktitle}{\emph{AAMAS}}. \bibinfo{pages}{2008--2010}.
\newblock


\bibitem[\protect\citeauthoryear{Smith}{Smith}{2023}]%
        {eth_stake_fees}
\bibfield{author}{\bibinfo{person}{Corwin Smith}.} \bibinfo{year}{2023}\natexlab{}.
\newblock \bibinfo{title}{PROOF-OF-STAKE REWARDS AND PENALTIES}.
\newblock \bibinfo{howpublished}{\url{ethereum.org/en/developers/docs/consensus-mechanisms/pos/rewards-and-penalties/}}.
\newblock
\newblock
\shownote{[Online].}


\bibitem[\protect\citeauthoryear{Tacchetti, Strouse, Garnelo, Graepel, and Bachrach}{Tacchetti et~al\mbox{.}}{2019}]%
        {tacchetti2019neural}
\bibfield{author}{\bibinfo{person}{Andrea Tacchetti}, \bibinfo{person}{DJ Strouse}, \bibinfo{person}{Marta Garnelo}, \bibinfo{person}{Thore Graepel}, {and} \bibinfo{person}{Yoram Bachrach}.} \bibinfo{year}{2019}\natexlab{}.
\newblock \showarticletitle{A neural architecture for designing truthful and efficient auctions}.
\newblock \bibinfo{journal}{\emph{arXiv preprint arXiv:1907.05181}} (\bibinfo{year}{2019}).
\newblock


\bibitem[\protect\citeauthoryear{Vickrey}{Vickrey}{1961}]%
        {Vickrey61}
\bibfield{author}{\bibinfo{person}{W. Vickrey}.} \bibinfo{year}{1961}\natexlab{}.
\newblock \showarticletitle{Counterspeculation, Auctions, and Competitive Sealed Tenders}.
\newblock \bibinfo{journal}{\emph{Journal of Finance}} \bibinfo{volume}{16}, \bibinfo{number}{1} (\bibinfo{date}{March} \bibinfo{year}{1961}), \bibinfo{pages}{8--37}.
\newblock


\bibitem[\protect\citeauthoryear{YCHARTS}{YCHARTS}{2022a}]%
        {ychartsBitcoin}
\bibfield{author}{\bibinfo{person}{YCHARTS}.} \bibinfo{year}{2022}\natexlab{a}.
\newblock \bibinfo{title}{Bitcoin Transactions Per Day}.
\newblock \bibinfo{howpublished}{\url{ycharts.com/indicators/bitcoin_transactions_per_day}}.
\newblock


\bibitem[\protect\citeauthoryear{YCHARTS}{YCHARTS}{2022b}]%
        {ychartsEth}
\bibfield{author}{\bibinfo{person}{YCHARTS}.} \bibinfo{year}{2022}\natexlab{b}.
\newblock \bibinfo{title}{Ethereum Transactions Per Day}.
\newblock \bibinfo{howpublished}{\url{ycharts.com/indicators/ethereum_transactions_per_day}}.
\newblock


\bibitem[\protect\citeauthoryear{Zhao, Chen, and Zhou}{Zhao et~al\mbox{.}}{2022}]%
        {zhao2022bayesian}
\bibfield{author}{\bibinfo{person}{Zishuo Zhao}, \bibinfo{person}{Xi Chen}, {and} \bibinfo{person}{Yuan Zhou}.} \bibinfo{year}{2022}\natexlab{}.
\newblock \showarticletitle{Bayesian-Nash-Incentive-Compatible Mechanism for Blockchain Transaction Fee Allocation}. In \bibinfo{booktitle}{\emph{Crypto Economics Security Conference (CESC)}}.
\newblock


\end{thebibliography}


\appendix

\section{Proofs for Results from Section 4 and 5~\label{app_proofs}}


\subsection{Proof of Theorem~\ref{thm::ideal_tfrm_imp}\label{app:thm_proof_ideal_tfrm}}

\begin{theorm}[\texttt{Ideal-TFRM} Impossibility]
If $r^\star$ is an anonymous rebate function that satisfies Theorem \ref{thm:order},  no \texttt{Ideal-TFRM} can guarantee a non-zero redistribution index (RI) in the worst case.
\end{theorm}
\begin{proof}
Consider the bid vector, $(v_1, \ldots, v_n)$ and a UIC rebate function $r_i^{\star} = g(v_{-i})$ as given by Theorem \ref{thm:order} satisfying Equation~\ref{eqn_ideal_tfrm}. W.l.o.g we assume $r_1^{\star} = g(v_2, \ldots, v_n) > 0$, that is agent 1 receives strictly positive rebate. Further $r_n^{\star} = g(v_1, \ldots, v_{n-1}) = 0$ since the last agent is not confirmed, and receives strictly zero rebate as per the constraint of Eq.~\ref{eqn_ideal_tfrm}. We now construct another bid vector $(v'_1, \ldots, v'_n)$ such that $v'_1 = v_{n}, v'_2 = v_1, v'_3 = v_4, \ldots, v'_{n} = v_{n-1}$. Under this bid profile, $r_1^{\star} = g(v'_2, \ldots, v'_n) = g(v_1, \ldots, v_{n-1}) = 0$, since for agent $n$, $r_n^* = 0$. Hence the worst case refund for the first agent will be $0$ contradicting our assumption that $r_1^{\star} > 0$. Similarly, we can construct bid vectors to show that for every agent $i\in C$, the worst case refund $r_i^* = 0$. 
\end{proof}

\subsection{Proof of Theorem~\ref{thm::imp_RRI}\label{app::thm_imp_rri}}

\begin{theorm}
     Given a strategic miner, it is impossible to design a TFRM with a linear rebate function that is RUIC, AE, both IR$_u$ and IR$_\mathbf{M}$ and guarantees a strictly positive RRI, i.e., $\hat{e}_{\textsf{wc}}>0$.
\end{theorm}
\begin{proof}
Consider selecting $k$ transactions from the total included $n$ with $n-k$ as price-setting transactions. Any deterministic linear rebate function that is RUIC must have the following form (from Theorem \ref{thm:order}):
$$ r_i = c_0 + c_1 b_1 + \ldots + b_i c_{i+1} + \ldots + c_{n-1} b_n $$
The above $r_i$s satisfy both IR$_u$ and IR$_\mathbf{M}$ only when the constants $c_i = 0, \forall i = {1, \ldots, k}$. The proof for this can be found in \cite[Claim 1]{Guo09}. Hence, in the rebate function, the top $k$ agents have the following form:
$$r_i = c_{k+1} b_{k+2} + \ldots + c_{n-1}b_n, \ \forall i\leq k$$
The rebate offered for the price-setting transactions is given by,
\begin{align*}
    r_{k+1} &= c_{k+1} b_{k+2} + \ldots + c_{n-1}b_n \\
    r_{k+2} &= c_{k+1} b_{k+1} + \ldots + c_{n-1}b_n \\
    \vdots\\
    r_{n} &= c_{k+1} b_{k+1} + \ldots + c_{n-1} b_{n-1}
\end{align*}
Now, consider the following miner deviation:
$$\hat{b}_{k+1} = b_k, \hat{b}_{k+2} = \ldots = \hat{b}_{n} = 0$$

The total VCG payments obtained is $k\hat{b}_{k+1} = k b_k$ and $r_i = 0, \forall i \leq k+1$. Hence the top $k+1$ transactions receive zero rebate. Now the miner has impersonated as $k+2, \ldots, n$ price setting users. Hence the miner receives back the rebate of $r_{k+2} = \ldots = r_{n} = c_{k+1} \hat{b}_{k+1} = c_{k+1} b_k$. Therefore the effective redistribution to the users is $\sum_{i=1}^k r_i = 0 \implies \hat{e}_{\textsf{wc}} = 0$.
\end{proof}

\section{Proofs for Results from Section 6~\label{app::section4}}

\subsection{Proof of Claim~\ref{claim:ir-nonir-tfrm}}

\begin{Claim}
A TFRM with $n$ included transactions and rebate $r$ is IR$_u$ if $r_n \ge 0$.
\end{Claim}
\begin{proof}
    Suppose there exists a bid vector $\mathbf{b}=(b_1, \ldots, b_n)$ and some user $i < n$ such that $r_i < 0$. $r_i$ is computed using $\mathbf{b}_{-i}$, i.e, $r_i = c_0 + c_1b_1+\ldots+c_{i-1}b_{i-1} + c_i b_{i+1} +\ldots+c_{n-1}b_n$. Now, consider a new bid vector $\mathbf{b}^{'}$ s.t., $b_1^{'} = b_2, b_2^{'} = b_3, \ldots, b_n^{'} = 0$. Observe that $\mathbf{b}^{'}_{-n} = \mathbf{b}_{-i}$, hence $r_i(\mathbf{b}_{-i},\cdot) =r_n(\mathbf{b}^{'}_{-n},\cdot)< 0$. Thus TFRM will be IR$_u$ if for any bid vector $r_n \ge 0$.
\end{proof}

\subsection{Proof of Claim~\ref{claim:12k-nonir-tfrm}}

\begin{Claim}
    If $c_0, \ldots, c_{n-1}$ satisfy IR$_u$ and Approx-IR$_\mathbf{M}$, then $c_i = 0$ for $i = 0, \ldots, k-1$.
\end{Claim}
\begin{proof}
    First, we show that $c_0 = 0$. Consider a bid vector $\hat{b}_i = 0$ for all $i$, then for IR$_u$ we must have $c_0  \geq 0$ (from Claim  \ref{claim:ir-nonir-tfrm}). To satisfy Approx-IR$_\mathbf{M}$, we must have $c_0 \leq 0$ hence $c_0 = 0$. If $c_i = 0,  \forall i$, we are done hence consider $j = \min\{i | c_i\neq 0\}$. Let $j < k$, then from IR$_u$ constraint we have that $r_n = c_0 + c_1 \hat{b}_1 + \ldots + c_{n-1}\hat{b}_{n-1} \geq 0$. Consider a bid vector s.t. $\hat{b}_i = 1$ for $i \leq j$ and $\hat{b}_i =0 $ for the rest. For this, $r_n = c_j$ therefore $c_j \geq 0$. While the Approx-IR$_\mathbf{M}$ constraint states that $r_1 + \ldots + r_n \leq k b_{k}$ for the above bid vector $\hat{b}$, we obtain $r_i = 0$ for $i\leq j$ as $c_j$ is multiplied with a bid that is $j^{th}$ highest hence 0. Whereas $r_i = c_j$ for $i>j$, because the $j^{th}$ highest bid is 1. Therefore the constraint can be re-written as $c_j (n-j) \leq k \hat{b}_k$. Because $j < k$, $\hat{b}_{k} = 0$ so the right hand side is 0.  Also $n-j > 0$ because $j < k < n$, therefore $c_j \leq 0$. But we already know from IR$_u$ that $c_j \geq 0$, hence $c_j =0$ for $j<k$. Hence by contradiction, $j\geq k$, i.e., $c_j = 0$ for $j < k$.
\end{proof}

\subsection{Proof of Claim~\ref{claim:iru-nonir-tfrm}}

\begin{Claim}
    The IR$_u$ constraint $r_n \ge 0$ and the worst-case fraction constraint (refer Figure~\ref{fig:tfrm-lp}) is equivalent to having $\sum_{j=k}^i c_j \geq f, \ i \in \{k, \ldots, n-1\}$.
\end{Claim}
\begin{proof}
    We consider the following lemma to prove the claim.
    \begin{Lemma}\cite[Lemma 1]{Guo07}
    \label{lemma:first}
    Given $n\in\mathbb{N}$ and set of real constants $s_1, \ldots, s_n$, $(s_1 t_1 + \ldots + s_nt_n \geq 0$ for any $t_1 \ge \ldots \ge t_n \geq 0)$ iff $\sum_{i=0}^j s_i \geq 0 $ for $j = 1,2 \ldots, n$.
    \end{Lemma}
    The proof of above lemma is given in \cite{Guo07}. From Claim \ref{claim:12k-nonir-tfrm}, we know that $c_i = 0 $ for $i < k$. Hence, we can re-write the worst-case constraint as follows,
    \begin{align*}
        \sum_{i=1}^k r_i &\geq f \cdot k b_{k+1} \\
        k \cdot (c_k b_{k+1} + \ldots + c_{n-1} b_n ) &\geq f \cdot k b_{k+1} \\
        (c_k - f) b_{k+1} + \ldots + c_{n-1} b_n & \geq 0
    \end{align*}
Given that $b_{k+1} \leq b_{k+2} + \ldots + b_n \geq 0$, we can invoke the above Lemma and thus, we obtain the following condition:
$c_k - f + c_{k+1} + \ldots + c_i \geq 0, \ \forall i = k , \ldots, n-1$. Thus, the worst-case fraction constraint can be written as $\sum_{j=k}^i c_j \geq f, i\in \{k, \ldots, n-1\}$.
\end{proof}

\subsection{Proof of Claim~\ref{claim:irmwc-nonir-tfrm}}

\begin{Claim}
    The Approx-IR$_\mathbf{M}$ constraint can be replaced by, 
    \begin{align*}
    (n-k)c_k \leq k~\mbox{~and~} &  n \cdot\sum_{j=k}^{n-1}c_j \leq k \\
 n\sum_{j=k}^{k+i-1}c_j + (n-k-i) c_{k+i} \leq k,~ & 
        i\in \{1, \ldots, n-k-1\}
    \end{align*}
\end{Claim}
\begin{proof}
We prove this using \cite[Lemma 1]{Guo07} on Approx-IR$_\mathbf{M}$ constraints.
  The Approx-IR$_\mathbf{M}$ constraint requires that $r_1 + \ldots + r_n \leq kb_k$ where $r_i = c_0 + c_1 b_1 + \ldots + c_i b_{i+1} + \ldots + c_{n-1} b_n$. Since $c_i = 0 , i<k$, the constraint is further simplified as follows,
  $$q_k b_k + q_{k+1} b_{k+1} + \ldots + q_{n} b_n \geq  0  $$
  where, $q_k = k - (n-k)c_k$ and $q_i = -(i-1)c_{i-1} - (n-i) c_i $, for $i =\{ k+1, \ldots, n-1\}$
  and $q_n = - (n-1) c_{n-1}$. By \cite[Lemma 1]{Guo08}, the Approx-IR$_\mathbf{M}$ constraint is equivalent to having $\sum_{i=k}^{j} q_i \geq 0$ for all $j =\{k, \ldots, n\}$. This can be simplified to obtain the following:
    \begin{align*}
        q_{k} \geq 0 &\iff (n-k)c_k \leq k \\
        q_{k} + \ldots + q_{m+i} \geq 0 &\iff n\sum_{j=k}^{k+i-1}c_j + (n-k-i) c_{k+i} \leq k, \\
        &\quad i\in \{1, \ldots, n-k-1\}\\
        q_{k} + \ldots + q_{n} \geq 0 &\iff n \sum_{j=k}^{n-1}c_j \leq k
    \end{align*}
This proves the claim.
\end{proof}

\subsection{Proof of Theorem~\ref{thm:nonir-tfrm}}

\begin{theorm}
For any $n$ and $k$ such that $n\ge k+2$, the \rtfrm mechanism is unique. The fraction redistributed to the top-k users in the worst-case is given by:
    $$f^{*} = \frac{k}{n}$$
    In \rtfrm, the rebate function is characterized by the following:
    $$c_k^{*} = \frac{k}{n},\  \ c_i^{*} = 0 \ \forall{i\neq k}$$
\end{theorm}
\begin{proof}
    We first show that $c_k^{*} = \frac{k}{n}$, $c_i^{*}=0, \forall i \neq k$ is a feasible solution to the LP given in Figure \ref{fig::nonir-tfrm}. 

    Note that, $\sum_{j=k}^i c_k^{*} = \frac{k}{n} \geq f^* = \frac{k}{n}, i\in\{k, \ldots, n-1\} $. Further $(n-k)c_k^{*} = (n-k)\frac{k}{n} < k$ as $n-k < n$. Observe that the third constraint $n \sum_{j=k}^{k+i-1} c_j^{*} + (n-k-i) c_{k+i}^{*} = nc_k^{*} = n\frac{k}{n}\leq k, i\in\{1, \ldots, n-k-1\}$ is also satisfied. Finally, $nc_k^{*} \leq k$ hence, $c^*$ and $f^*$ is a feasible solution to the LP.

    Now we show that if there exists any solution $\hat{c}, \hat{f}$ that satisfies the constraints of the LP, then $\hat{c} = c^{*}$ and $\hat{f} = f^*$. First, let $x_j = \sum_{i=k}^j c_i$ for $j= k,\ldots, n-1$. The LP constraints can be re-written as the following,
    \begin{align*}
        x_j &\geq f,  \ j=\{k, \ldots, n-1\} \\
        (n-k)x_k &\leq k \\
        (k+i) x_{k+i-1} + (n-k-i) x_{k+i} &\leq k \ i = \{1,  \ldots, n-k-1\}\\
        nx_{n-1} &\leq k 
    \end{align*}
    We know that $x_k^* = f^*$, from first constraint we have $\hat{x}_k \geq \hat{f} \geq f^* = x_k^*$, i.e., $\hat{x}_j \geq x_j^*$ for $j=\{k, \ldots, n-1\}$. From second constraint we have that, $x_k^* \leq $ From the third constraint, we obtain, for i = 1, $(k+1) x_k^* + (n-k-1) x_{k+1}^* = k$ and for any $\hat{x}$, $(n-k-1) \hat{x}_{k+1} \leq k - (k+1) \hat{x}_{k} \leq k - (k+1) x_k^* = (n-k-1) x_{k+1}^*$, i.e., $\hat{x}_{k+1} \leq x_{k+1}^*$. Similarly we can obtain $\hat{x}_k \leq x_k^*$ further using $i={2, \ldots, n-k-1}$ and the fourth constraint we obtain $\hat{x_j} \leq x_j^{*}$, for $j =k+2, \ldots, n-1$. From the first and fourth constraint we can conclude that $\hat{x}_j = x_j^*$ for $j = k, \ldots, n-1$. This completes the proof of the theorem.
\end{proof}

\subsection{Proof of Theorem~\ref{thm:rri-nonir-tfrm}}

\begin{theorm}
  Consider $n$ included transactions with the set $C$ as confirmed transactions such that $|C|=k$ with the remaining $n-k$ as price-setting transactions. Irrespective of any miner manipulation, \rtfrm\ ensures strictly positive RRI or $ \hat{e}_{\textsf{wc}} =  c_k^*$.
\end{theorm}
\begin{proof}
The VCG payments to the miner and rebates both depend on the unconfirmed transaction $b_{k+1}$. Manipulating $b_{k+1}$ would change both the payments and refund in a way that the fraction of redistribution remains constant.
    From Definition \ref{def:rri}, we know that,
    \begin{align*}
        \hat{e}_{\textsf{wc}} &=  inf_{\hat{\mathbf{b}}} \frac{\sum_{i\in S} r_i}{p(\hat{\mathbf{b}})} \\
          &=  inf_{\hat{\mathbf{b}}} \frac{k(c_k \hat{b}_{k+1})}{k \hat{b}_{k+1}} \\
         & \quad \left(\mbox{As } S= \{1,\ldots,k\} \ \mbox{and}\  p(\hat{\mathbf{b}}) \mbox{ is VCG}\right)  \\
         & \quad \left(r_i \mbox{ given by Theorem \ref{thm:nonir-tfrm}}\right)\\  \\
         &= c_k = \frac{k}{n}
    \end{align*}
    
    This proves the theorem.
\end{proof}

\section{Proofs for results from Section 7~\label{app::sec7}}

\subsection{Proof of Theorem \ref{thm:rrtfrm}}
\label{app:thm_rrtfrm}

\begin{theorm}
    For any $n$ and $k$ such that $n\geq k+2$ and any bid profile $\mathbf{b} = (b_1, \ldots, b_n)$, and probability $\alpha \in (0,1)$ \rrtfrm\ has an expected redistribution fraction $f^* = \alpha \cdot \frac{k}{n}$. Further it satisfies $AE$, $RUIC$, $IR_u$, and is $IR_\mathbf{M}$ when $\alpha \leq \frac{n}{k+ (n-k)b_k/b_{k+1}} $. 
\end{theorm}

\begin{proof}
\rrtfrm\ is randomized \rtfrm, where randomness is introduced due to the $\alpha$ parameter. 

\smallskip

\noindent\rrtfrm\ is $AE$ as the $k$ highest bids are allotted the slots according to the confirmation rule (Figure \ref{fig:RRTFRM_framework}). 

\smallskip

\noindent\rrtfrm\ is $IR_u$ as utility is non-negative for each included transaction. For the confirmed transaction $i \in \{1, \ldots, k\}$, utility is 
\begin{equation*}
        u_i = 
        \begin{cases}
           \theta_i - \left(\frac{n-k}{n}\right)b_{k+1} \geq b_i - \left(\frac{n-k}{n}\right)b_{k+1} \geq 0 \quad \mbox{w.p.}\ \  \alpha \\
            \theta_i - b_{k+1} \geq b_i - b_{k+1} \geq 0 \quad \mbox{w.p.}\  \ (1-\alpha)
        \end{cases}
    \end{equation*}
 as $\theta_i \geq b_i$ and $b_i \geq b_{k+1}$. Similary for the included but non-confirmed transaction, $j \in \{k+1, \ldots, n\}$, the utility is $u_j = \frac{k}{n} b_k \geq 0$ w.p. $\alpha$ or $u_j = 0$.

\smallskip
\noindent\rrtfrm\ is $RUIC$-in-expectation. For any confirmed transaction $i \in \{1, \ldots, k\}$ expected utility is $ u_i = \theta_i - \left(\frac{n-\alpha k}{n}\right) b_{k+1}$. The user may submit a dishonest bid $b_i*'$ s.t., i) $b'_i > \theta_i$ or ii) $b'_i < \theta_i$.

\noindent i) When $b'_i > \theta_i$, the bid is still confirmed and the expected utility is unchanged.

\noindent ii) When $b'_i < \theta_i$, s.t. $b'_i > b_{k+1}$, the transaction remains confirmed and the expected utility is unchanged. Thus, consider $b'_i < b_{k+1}$, then the transaction is not confirmed and now $b_k = b_{k+1}$ therefore the expected utility to $i$ is $\alpha \frac{k}{n} b_{k+1} \leq u_i$. Hence no confirmed user gains in expected upon underbidding.

For any transaction i.e., included but not confirmed, $j \in \{k+1, \ldots, n\}$, the expected utility is $u_i = \alpha \frac{k}{n}b_k $ and we have the same two cases as above,

\noindent i) When $b'_j > \theta_j$, s.t., $b'_j < b_k$ the transaction remains unconfirmed, and the expected utility is unchanged. Thus, consider $b'_j \geq b_k$, the transaction gets confirmed and now, $b_{k+1} = b_k \geq \theta_j$. Thus the expected utility of $j$ is given by $\theta_j - b_k + \alpha \frac{k}{n} b_{k} \leq u_j $. Hence the user does not gain in expectation upon overbidding.

\noindent ii) When $b'_j < \theta$, the user remains non-confirmed, and the maximum expected utility remains unchanged.

\rrtfrm\ satisfies $AE$, $IR_u$ and $RUIC$-in-expectation.

\smallskip

\noindent\rrtfrm\ is $IR_\mathbf{M}$-in-expectation. The payment obtained by an honest miner is given by $kb_{k+1}$. The expected rebate paid by the honest miner is given by, $\alpha \left[k \frac{k}{n} b_{k+1} + (n-k) \frac{k}{n} b_k\right] $. If $\alpha = 1$, the refund is equal to the refund in \rtfrm\ given by Equation \ref{eq:rebatenonirtfrm}. In order to ensure $IR_\mathbf{M}$ for the honest miner, the following must be true,

    \begin{align*}
        kb_{k+1} &\geq \alpha \left[k \frac{k}{n} b_{k+1} + (n-k) \frac{k}{n} b_k\right] \\
        b_{k+1} &\geq \alpha \left[\frac{k}{n} b_{k+1} + \frac{(n-k)}{n}b_k \right]\\
        \alpha &\leq \frac{nb_{k+1}}{ k b_{k+1} + (n-k)b_k}
    \end{align*}
Therefore \rrtfrm\ satisfies $IR_\mathbf{M}$ when $\alpha (b) =  \frac{nb_{k+1}}{ k b_{k+1} + (n-k)b_k} $ and total rebate is $k b_{k+1}$. 
\end{proof}

\subsection{Proof of Theorem~\ref{thm:rri-rr-tfrm}}

\begin{theorm}
  Consider $n$ included transactions with the set $C$ as confirmed transactions such that $|C|=k$ with the remaining $n-k$ as price-setting transactions. With miner manipulation \rrtfrm\ ensures: expected RRI or $ \mathbf{E}[\hat{e}_{\textsf{wc}}] =  \alpha \frac{k}{n}$.
\end{theorm}
\begin{proof}
    Similar to Theorem \ref{thm:rri-nonir-tfrm}, the fraction of redistribution remains constant. For every true user (not fake), the $\alpha k/n$ fraction of the payment is returned back as the rebate in expectation.
\end{proof}


\end{document}